\documentclass[11pt,reqno]{article}

\usepackage[T1]{fontenc}
\usepackage[utf8]{inputenc}
\usepackage[a4paper,margin=1in]{geometry}
\usepackage{microtype}
\usepackage{newpxtext,newpxmath} 
 
\usepackage{amsmath}
\usepackage{amssymb}

\usepackage{amsthm}
\usepackage{mathtools}
\usepackage{esint}
\usepackage{extarrows}
\usepackage{accents}
\usepackage{tensor}
\usepackage{nicefrac}
\usepackage{nccmath}
\numberwithin{equation}{section}
\allowdisplaybreaks
\usepackage{hyperref} 
\usepackage[nameinlink,noabbrev]{cleveref}

\crefname{assumption}{assumption}{assumptions}
\Crefname{assumption}{Assumption}{Assumptions}

\usepackage{booktabs,tabularx,array}

\usepackage{enumitem}
\setlist[itemize]{leftmargin=2.2em}
\setlist[enumerate]{leftmargin=2.2em}

\providecommand{\mathscr}{\mathcal}

\theoremstyle{plain}
\newtheorem{theorem}{Theorem}[section]
\newtheorem{proposition}[theorem]{Proposition}
\newtheorem{lemma}[theorem]{Lemma}
\newtheorem{corollary}[theorem]{Corollary}

\theoremstyle{definition}

\newtheorem{assumption}[theorem]{Assumption}
\newtheorem{remark}[theorem]{Remark}
\newtheorem{example}[theorem]{Example}
\newtheorem{problem}[theorem]{Problem}

\DeclareMathOperator{\KL}{KL} 
\newcommand{\R}{\mathbb{R}}

\newcommand{\E}{\mathbb{E}}
\newcommand{\Prob}{\mathbb{P}}

\newcommand{\Ptwo}{\mathcal{P}_2}

\newcommand{\Var}{\mathrm{Var}}

\title{\LARGE \textbf{Differential Beliefs in Financial Markets Under Information Constraints: \textcolor{black}{A Modeling Perspective}}}

\author{
\textsc{Karen Grigorian}\thanks{Department of Statistics and Applied Probability, University of California Santa Barbara; grigorian@ucsb.edu}
\and
\textsc{Robert A. Jarrow}\thanks{Samuel Curtis Johnson Graduate School of Management, Cornell University; robertjarrow@cornell.edu}
}

\date{\today}

\begin{document}
\maketitle

\begin{abstract}
We apply the theory of McKean-Vlasov-type SDEs to study several problems related to market efficiency in the context of partial information and partially observable financial markets: (i) convergence of reduced-information market price processes to the true price process under an increasing information flow; (ii) a specific mechanism of shrinking biases under increasing information flows; (iii) optimal aggregation of expert opinions by a trader seeking a positive alpha. All these problems are studied by means of (conditional) McKean-Vlasov-type SDEs, Wasserstein barycenters, KL divergence and relevant tools from convex optimization, optimal control and nonlinear filtering. We supply the theoretical results in (i)-(iii) with concrete simulations demonstrating how the proposed models can be applied in practice to model financial markets under information constraints and the arbitrage-seeking behavior of traders with differential beliefs. 
\end{abstract}


\begin{color}{black}
\section{Introduction}

Our general goal is to study the interplay between the notions of market efficiency, partial information and partial observability from a modeling perspective. This is accomplished in a sequence of three successively more specialized models which, as we show, are highly amenable to practical implementation and simulation.  

The original inspiration for some of the suggested models came from several examples of model uncertainty presented in \cite{OksendalSulem} and the interesting application of the notion of barycenters in \cite{JaimungalPesenti:KLbarycentre}, as well as the authors' own published research \cite{GKJR1, GKJR2, GKJR3, GKJR4, GKJR5, GKJR6, JarrowLarsson} which naturally suggested a deeper study of financial markets under information constraints.

In the first application, the most general setting, we propose a model
based on a McKean-Vlasov-type stochastic differential equation (MKVSDE)
with a barycentric measure input that explicitly describes how increasing
information flows impact stock prices and market efficiency. Here,
we first construct a hypothetical market with price $S$ that satisfies
No Free Lunch with Vanishing Risk (NFLVR) and No Dominance (ND) with
respect to the information set $\mathbb{F}$, which represents all
the information available in the market. These properties imply that
there exists an equivalent martingale measure (EMM) with respect to
$\mathbb{F}$ in the hypothetical market, see \cite{Jarrow} for the
definitions of NFLVR, ND, and EMM and the stated theorems. Using the
definition of an efficient market as in Jarrow and Larsson \cite{JarrowLarsson},
by construction, this hypothetical market is efficient with respect
to $\mathbb{F}$ in the sense of strong-form efficiency \cite{Jarrow},
Ch.16.

Next, we consider a sequence of actual markets, indexed by $n$. The
markets differ by the private information available to the $m$ traders,
indexed by $i$, and denoted $\mathbb{G}^{i,n}\subset\mathbb{F}$.
We assume that the \textit{true} price process \textbf{$S$} is not
observable in the actual market, hence it is not $\mathbb{G}^{i,n}$-adapted. We call it the true market price because it is the market
price that would exist if all the traders knew $\mathbb{F}$. For
the actual market, the total information available (in the sense of strong-form
efficiency) is $\mathbb{H}^{n}:=\bigvee_{i=1}^{m}\mathbb{G}^{i,n}$.
The market price for the stock in the actual market is denoted $\tilde{S}^{n}$.

For the actual market with price $\tilde{S}^{n}$, we no longer assume
that there exist an EMM $\mathbb{\tilde{Q}}$, hence the actual market
may violate either NFLVR or ND. It follows that the actual market
may be inefficient with respect to $\mathbb H^n$. We show that
as the information increases in a market, i.e. $\mathbb{H}^{n}\uparrow\mathbb{F}$,
the actual markets approach the hypothetical market that satisfies
NFLVR and ND, i.e. they approach an efficient market with respect
to $\mathbb{F}$.

The second and third applications studied are motivated by the observation
that in financial markets, traders (hedge funds, investment and commercial
banks, pension funds, insurance companies) use different factor models
to identify mispriced securities looking for arbitrage opportunities
(see Jarrow and Protter \cite{JarrowProtter}), this is called \textquotedblleft the
search for positive alphas.\textquotedblright{} The search for positive
alphas is the search for a security's ``true'' drift (expected return). 

For the second application, using the same market set-up as in the
first application, the $n$th market may not satisfy NFLVR or ND
with respect to $\mathbb{G}^{n}$. Here, we represent the search for
positive alphas by a single trader using their private information
to estimate the drift of the actual market price. We study how increasing
this information impacts the market price process. We show that the
search for a positive alpha removes FLVR and dominated assets in the
market, thereby increasing market efficiency.

Finally, in the third application, we study an optimal aggregation problem where a trader seeking arbitrage opportunities (positive alphas) is acting under information constraints $\mathbb G \subset \mathbb F$ and consults a (possibly continuous) set of experts $\Lambda$, who propose corrections to the observed drift, represented by expert-indexed random field $\rho^\lambda$, which the trader wishes to aggregate to obtain the best estimate of the unknown true drift $a$. The trader tries to minimize the distance between the aggregated correction term $\bar\rho$ and the trader's own estimate $\widehat a$ based on available information $\mathbb G$, taking into account their own prior beliefs on the expert community, represented by a flow $\pi$ of probability measures over $\Lambda$. Here as well, we have that the hypothetical market with price $S$ satisfies NFLVR and
ND and is efficient with respect to $\mathbb{F}$.
As in the preceding case, we now no longer assume that there exist an EMM $\mathbb{\tilde{Q}}$ for the market price $\tilde{S}$ in the actual market, hence the actual market may violate either NFLVR or ND. It follows that the actual market may be inefficient with respect to $\mathbb{G}$.

This analysis begins from purely financial arguments, and eventually arrives at well-known variational formulas for KL divergence, exponential tilting, and KL-regularized decision/control. 
The Gibbs measure form of the optimizer follows from the variational characterization of log-partition functions and the Donsker–Varadhan (DV) variational principle for relative entropy; see standard treatments of large deviations \cite{WainwrightJordan2008,DemboZeitouni,DupuisEllis,CoverThomas}.
The geometry of $I$–projections (KL projections) onto convex sets under linear constraints is classical and gives exponential-family solutions closely related to our characterization, see \cite{Csiszar1975}.

KL-based distributionally robust optimization (DRO) provides another close analogue: worst-case expectations over a KL-ball ambiguity set admit dual solutions that are exponential tilts of the nominal distribution, similar to our optimal Gibbs measure-valued controls and KL constraints, see \cite{HuHong2013}.

Aggregation and opinion pooling also often lead to exponential reweighting. The logarithmic opinion pool and its characterizations (and weighting schemes) provide aggregation rules that formally coincide with our optimal Gibbs measures, but most commonly use KL-barycetners, see, e.g., \cite{Heskes, GenestMcConwaySchervish}. However, the core ideas of this paper are inherently financial in nature and admit very explicit interpretation in the context of partial information and partially observable financial markets. We end each section with explicit simulations showing how the proposed models can be applied in finance.

    We acknowledge the use of ChatGPT 5 Pro in obtaining the code for simulations as well as in testing out a large number of model hypotheses which often led to dead ends, and hence expediting the creative part of research. It helped with transferring handwritten notes to LaTex, and finding connections to information theory. It also suggested proof strategies and arguments in several technical lemmas and propositions that eventually worked after our guided fixing and helping it identify its own errors. While in most of these cases we continually steered it away from erroneous suggestions and arguments, it was very helpful nevertheless and saved from hours of cumbersome work. 
    
The paper is structured as follows: section 2 introduces a general McKean-Vlasov-type SDE-based model of price dynamics affected by traders' beliefs and studies the convergence of a sequence of reduced-information markets to the market with full information; section 3 proposes a specialized model to analyze the evolution of individual biases under increasing information; section 4 investigates the problem of optimal aggregation of experts' opinions by a trader seeking to find an arbitrage opportunity or a dominated asset; section 5 concludes, and the appendix contains some well-known facts about the measurability properties of flows of probability kernels. 
\end{color}

\section{\textcolor{black}{Differential Beliefs and Convergence to an Efficient Market}}

\subsection{Preliminaries and Model Set-Up}

We study a financial market supported on some complete filtered probability space $(\Omega,\mathcal F,\mathbb F, \mathbb P)$ endowed with a filtration $\mathbb F=(\mathcal F_t)_{t\in[0,T]}$ for a fixed finite horizon $T>0$ which satisfies the usual conditions. Let $W$ be a $D$--dimensional $\mathbb F$--Brownian motion. The $d$--dimensional \emph{true price} process $S$ is the $\mathbb F$--adapted Markov diffusion
\begin{equation}\label{eq:trueS}
dS_t=\bar b(t,S_t)\,dt+\bar\sigma(t,S_t)\,dW_t,\qquad S_0\in L^2,\qquad
\mathbb E\Big[\sup_{t\le T}|S_t|^2\Big]<\infty,
\end{equation}
where $\bar b:[0,T]\times\mathbb R^d\to\mathbb R^d$ and $\bar\sigma:[0,T]\times\mathbb R^d\to\mathbb R^{d\times D}$ satisfy the usual global Lipschitz and linear--growth conditions ensuring well--posedness.

The true price process $S$ is not fully observable and each of the $m$ \textcolor{black}{traders} has access to \textcolor{black}{their respective} information flow defined by a right--continuous, complete subfiltration $\mathbb G^{i,n}=(\mathcal G^{i,n}_t)_{t\in[0,T]}$ of $\mathbb F$, \textcolor{black}{$i\in\{1,\dots,m\}, n\in\mathbb N$ fixed,} which \textcolor{black}{determines} their individual opinions/proposals on the drift and volatility, assumed to be some functions of the flows $\pi^i$ of conditional distributions of the price process given their respective information flow. \textcolor{black}{Thus, for each $(i,n)$ and $t\in[0,T]$, the \emph{i}th \emph{trader's beliefs} are given by 
\[
\pi_{t}^{i,n}:=\mathcal{L}(S_{t}\mid\mathcal{G}_{t}^{i,n}),
\]
the conditional law of $S_t$ given $\mathcal G^{i,n}_t$, i.e. a $\mathcal P_2(\mathbb R^d)$-valued random variable, where $\mathcal P_2(\mathbb R^d)$ is the set of all probability measures on $\mathbb R^d$ with finite second moments.} 

\textcolor{black}{We stress that the traders have different conditional beliefs, given their information, and the conditional law is based on the statistical probability measure $\mathbb{P}$. When only drifts are affected by the views, this is not a serious restriction and, with obvious modifications of the arguments below, the results can easily be generalized to the case when the the individual conditional laws also depend on the traders' probability measures
$\mathbb{P}^{i}$ that are equivalent to $\mathbb{P}$, and hence the drifts can be changed by Girsanov's theorem. When volatilities are affected, the argument has to be refined, see \cite{JaimungalPesenti:KLbarycentre} for one approach.}

\begin{color}{black}
The market combines their respective beliefs into a single price process $\tilde S$, whose drift and volatility depend on the barycenter of the $m$ traders' views, denoted by $\tilde \pi$. Hence, the total available information in the market is 
\[
\mathbb{H}^{n}:=\bigvee_{i=1}^{m}\mathbb{G}^{i,n},\qquad\mathcal{H}_{t}^{n}:=\bigvee_{i=1}^{m}\mathcal{G}_{t}^{i,n}.
\]

\medskip{}

The aggregate belief of the market, called the \emph{market beliefs},
is defined by

\begin{equation}
\tilde{\pi}_{t}^{\,n}\in\text{argmin}_{\nu\in\mathcal{P}_{2}(\mathbb{R}^{d})}\ \sum_{i=1}^{m}w_{i,t}^{(n)}\,W_{2}^{2}\big(\nu,\pi_{t}^{i,n}\big),\qquad w_{i,t}^{(n)}>0, \qquad t\in[0,T],\label{eq:tildepi}
\end{equation}
where $W_2$ is the 2-Wasserstein distance between two probability measures. Thus, \eqref{eq:tildepi} is the standard Wasserstein barycenter of probability measures, see \cite{AguehCarlier} for details on existence and properties. Existence of minimizers for \eqref{eq:tildepi} is classical for a finite family of inputs in $(\mathcal P_2(\mathbb R^d),W_2)$, and, assuming the individual flows of measures are $\mathbb G^{i,n}$-progressively measurable, by a measurable selection theorem one may choose $(t,\omega)\mapsto\tilde\pi^{\,n}_t(\omega)$ to be $\mathbb H^n$--progressively measurable. The argument is nontrivial, but standard and rests on checking the properties of normal integrands and using a measurable selection theorem, as in \cite{RockafellarWets}, Thm.~14.37. Thus, our market beliefs are represented by the minimizer of a weighted average of $W_2$-distances to the traders' beliefs. The weights $w_{i,t}^{(n)}>0$
imply that each trader has a positive impact on the market beliefs. This condition and the technical assumption of at least one of the measures $\pi_{t}^{i,n}$ being absolutely continuous wrt the Lebesgue measure ensure the uniqueness of the barycenter measure $\tilde{\pi}_{t}^{\,n}$, \cite{AguehCarlier} Prop. 3.5. Our analysis, however, does not rely on uniqueness and only requires the existence of one such measure.

The weight selection mechanism is not explicitly given. However, intuitively
it is generated by each trader's market impact on prices either through
the magnitude of the trader's trade size or the trader's influence
on the market via media communication and online followers (e.g. Warren
Buffet).

\medskip{}

We define the \emph{market price }to be 
\begin{equation}
d\tilde{S}_{t}^{(n)}=b\big(t,\tilde{S}_{t}^{(n)},\tilde{\pi}_{t}^{\,n}\big)\,dt+\sigma\big(t,\tilde{S}_{t}^{(n)},\tilde{\pi}_{t}^{\,n}\big)\,dW_{t},\qquad\tilde{S}_{0}^{(n)}=S_{0},\label{eq:tildeS}
\end{equation}
driven by the \emph{same} Brownian motion $W$. This assumption is not restrictive and will be relaxed in the subsequent sections. We adopt it in this section to focus our analysis on the pure impact of differential beliefs.  We assume the Lipschitz and linear--growth conditions
\begin{equation}\label{eq:coeff-Lip}
\begin{aligned}
&\exists L\ge 1\ \text{s.t.}\ \forall t\in[0,T],\ x,y\in\mathbb R^d,\ \mu,\nu\in\mathcal P_2(\mathbb R^d):\\
&\quad |b(t,x,\mu)-b(t,y,\nu)| + \|\sigma(t,x,\mu)-\sigma(t,y,\nu)\|
\ \le\ L\big(|x-y|+W_2(\mu,\nu)\big),\\
&\quad |b(t,x,\mu)|^2+\|\sigma(t,x,\mu)\|^2\ \le\ L\big(1+|x|^2+\int|z|^2\,\mu(dz)\big),
\end{aligned}
\end{equation}
together with the \emph{compatibility condition} linking \eqref{eq:trueS} and \eqref{eq:tildeS}:
\begin{equation}\label{eq:compat}
\bar b(t,x)=b\big(t,x,\delta_x\big),\qquad \bar\sigma(t,x)=\sigma\big(t,x,\delta_x\big),\qquad \forall (t,x)\in[0,T]\times\mathbb R^d.
\end{equation}
Under \eqref{eq:coeff-Lip}, \eqref{eq:tildeS} is well--posed and $\mathbb E\sup_{t\le T}|\tilde S^{(n)}_t|^2<\infty$.

 The true price is the market price that would exist if all the traders knew \textbf{$\mathbb{F}$.} \textbf{$\mathbb{F}$ }represents all the available information in the sense of strong form efficiency \cite{Jarrow}, Ch.16.

We consider a sequence of markets, indexed by $n$. The markets differ by the information available to the $m$ traders, indexed by $i$ with private information $\mathbb{G}^{i,n}\subset\mathbb{F}$. We want to study how increasing the information available to traders impacts market efficiency, i.e. the convergence of prices as $n\rightarrow\infty$.

The mechanism that generates the market price based on the trader's
beliefs and trading strategies is outside the model's structure. This
is a ``reduced form'' model. This contrasts with a ``structural
model'' that determines the market price given a specification of
each trader's endowments, portfolio and consumption optimization problem,
and market clearing mechanism. In the classical asset pricing literature,
this is given by a Radner equilibrium. In the market microstructure
literature, this would be based on a Nash equilibrium.

We note that our reduced form specification of the price process is consistent with these structural models, and possibly other market clearing mechanisms. However, the converse is also true. The
market price process need not be an equilibrium price with respect
to the standard paradigms mentioned above.

We assume that the filtration generated by $\tilde{S}_{t}^{(n)}$
is \textcolor{black}{contained in $\mathbb{H}^{n}$}. Thus, we have constructed two markets: a hypothetical market with price $S$ and
the actual market with price $\tilde{S}$. We assume that these markets
have the standard asset pricing structure, trading strategies, etc.
as in \cite{JarrowLarsson}. We assume that there exists an equivalent martingale measure
(EMM) $\mathbb{Q}$ for the true price $S$ in the hypothetical market
constructed above. This hypothetical market satisfies No Free Lunch with Vanishing Risk (NFLVR) and No Dominance (ND), see \cite{Jarrow}.
The EMM need not be unique, so the hypothetical market can be incomplete.

Using the definition of an efficient market in Jarrow and Larsson
\cite{JarrowLarsson}, the hypothetical market is efficient with respect to $\mathbb{F}$, i.e. it is strong-form efficient.
It is also efficient with respect to smaller information sets, so
it is semi-strong form and weak-form efficient as well. Hence, it is the ideal market.

We do not assume that there exists an EMM $\mathbb{\tilde{Q}}$
for the market price $\tilde{S}$ in the actual market constructed
above. Hence, the actual market may violate either NFLVR or ND. If it exists, the EMM need not
be unique, so the actual market can be incomplete. By the definition
of an efficient market in \cite{JarrowLarsson}, the actual
market may be inefficient with respect to $\mathbb{H}^{n}$, and is efficient
with respect to $\mathbb{F}$ if for some finite $n$
\[
\tilde{S}^{(n)}=S.
\]

In our case, the actual market may be inefficient with respect to $\mathbb{H}^{n}$, and is inefficient with respect to $\mathbb{F}$. In the actual market, the information from $\mathbb{F}$ could generate
arbitrage opportunities, as discussed in \cite{GKJR6}. This issue in studied in section 2.
We also note that if there is an EMM in the actual market and one trader $i$ for which \textbf{$\mathbb{G}^{i,n}=\mathbb{F}$},
then, because $\mathbb{H}^{n}$ includes $\mathbb{G}^{i,n}=\mathbb{F}$,
the $n$th market is efficient with respect to $\mathbb{F}$ immediately. This implies the interesting structure is where no individual trader knows $\mathbb{F}$, therefore we assume that
\[
\mathbb{G}^{i,n}\subsetneq\mathbb{F}
\]
 for all $i$ and all $n$.

\medskip{}

\end{color}

In this section we do not posit a specific functional form for the individual/combined proposed drift and volatility, and their explicit dependence on traders' beliefs. In the subsequent sections we study more specialized models where these dependencies are given explicitly. Finaly, we provide explicit simulations of our theoretical results and demonstrate how they can be implemented in practice.

\medskip

We shall repeatedly use the following stability estimate; its proof is standard and included for completeness.

\begin{lemma}\label{lem:MV-stability}
Assume \eqref{eq:coeff-Lip} and \eqref{eq:compat}. Then for each $n$ and all $t\in[0,T]$,
\begin{equation}\label{eq:MV-stability}
\mathbb E\Big[\sup_{s\le t}\big|\tilde S^{(n)}_s-S_s\big|^2\Big]
\ \le\ C_{L,T}\int_0^t \mathbb E\,W_2^2\big(\tilde\pi^{\,n}_u,\delta_{S_u}\big)\,du,
\end{equation}
for a constant $C_{L,T}<\infty$ depending only on $L$ and $T$.
\end{lemma}

\begin{proof}
    Let $\Delta_t:=\tilde S^{(n)}_t-S_t$. Using \eqref{eq:tildeS}, \eqref{eq:trueS}, and the compatibility \eqref{eq:compat},
\[
d\Delta_t=\Big(b\big(t,\tilde S^{(n)}_t,\tilde\pi^{\,n}_t\big)-b\big(t,S_t,\delta_{S_t}\big)\Big)\,dt
+\Big(\sigma\big(t,\tilde S^{(n)}_t,\tilde\pi^{\,n}_t\big)-\sigma\big(t,S_t,\delta_{S_t}\big)\Big)\,dW_t,\qquad \Delta_0=0.
\]
By Burkholder--Davis--Gundy, Jensen, and \eqref{eq:coeff-Lip}, for some $C=C_{L,T}$,
\begin{align*}
\mathbb E\Big[\sup_{s\le t}|\Delta_s|^2\Big]
&\le C\,\mathbb E\int_0^t \Big|b\big(u,\tilde S^{(n)}_u,\tilde\pi^{\,n}_u\big)-b\big(u,S_u,\delta_{S_u}\big)\Big|^2\,du\\
&\quad + C\,\mathbb E\int_0^t \big\|\sigma\big(u,\tilde S^{(n)}_u,\tilde\pi^{\,n}_u\big)-\sigma\big(u,S_u,\delta_{S_u}\big)\big\|^2\,du\\
&\le C\,\mathbb E\int_0^t \Big(|\Delta_u|^2+W_2^2(\tilde\pi^{\,n}_u,\delta_{S_u})\Big)\,du.
\end{align*}
Gronwall's lemma yields
\[
\mathbb E\Big[\sup_{s\le t}|\Delta_s|^2\Big]\ \le\ C_{L,T}\int_0^t \mathbb E\,W_2^2\big(\tilde\pi^{\,n}_u,\delta_{S_u}\big)\,du,
\]
as claimed.
\end{proof}

We also record a basic identity and a simple domination bound that will be used repeatedly.

\begin{lemma}\label{lem:W2-identity}
Let $X\in L^2$ and $\mathcal G\subset\mathcal F$ be a sub-$\sigma$-algebra. Then
\[
W_2^2\big(\mathcal{L}(X\mid\mathcal G),\delta_X\big)
=\int |y-X|^2\,\mathcal{L}(X\mid\mathcal G)(dy)
=\mathrm{Var}(X\mid \mathcal G)+\big|X-\mathbb E[X\mid \mathcal G]\big|^2.
\]
In particular, $\mathbb E\,W_2^2\big(\mathcal{L}(X\mid\mathcal G),\delta_X\big)=2\,\mathbb E\big|X-\mathbb E[X\mid \mathcal G]\big|^2\le 2\,\mathbb E|X|^2$.
\end{lemma}

\begin{proof}
    Standard. 
\end{proof}

\begin{lemma}\label{lem:bar-dom}
Let $\mu_1,\dots,\mu_m\in\mathcal P_2(\mathbb R^d)$, $w\in\Delta_m$, and let $\mathrm{Bar}_w(\mu_1,\dots,\mu_m)$ be any minimizer in \eqref{eq:tildepi}. Then, for every $\rho\in\mathcal P_2(\mathbb R^d)$,
\[
W_2^2\big(\mathrm{Bar}_w(\mu_1,\dots,\mu_m),\rho\big)\ \le\ 4\,\max_{1\le i\le m} W_2^2(\mu_i,\rho).
\]
\end{lemma}

\begin{proof}
Let $\bar\mu:=\mathrm{Bar}_w(\mu_1,\dots,\mu_m)$. Pick $i^\star$ minimizing $W_2(\bar\mu,\mu_i)$. By optimality of $\bar\mu$,
\[
\sum_{i=1}^m w_i\,W_2^2(\bar\mu,\mu_i)\ \le\ \sum_{i=1}^m w_i\,W_2^2(\rho,\mu_i)\ \le\ \max_i W_2^2(\rho,\mu_i).
\]
Hence $W_2^2(\bar\mu,\mu_{i^\star})\le \max_i W_2^2(\rho,\mu_i)$. By the triangle inequality and $(a+b)^2\le 2(a^2+b^2)$,
\[
W_2^2(\bar\mu,\rho)\ \le\ 2\,W_2^2(\bar\mu,\mu_{i^\star})+2\,W_2^2(\mu_{i^\star},\rho)\ \le\ 4\,\max_i W_2^2(\mu_i,\rho).\qedhere
\]
\end{proof}

\subsection{Convergence Under Uniformly Increasing Information}

The following assumption plays a key role in ensuring convergence to the true price process.

\begin{assumption}\label{ass:all-improve}
For each $i\in\{1,\dots,m\}$ and each $t\in[0,T]$, the $\sigma$--algebras increase in $n$ and exhaust $\mathcal F_t$:
\[
\mathcal G^{i,1}_t\subseteq \mathcal G^{i,2}_t\subseteq\cdots,\qquad \sigma\Big(\bigcup_{n\ge 1}\mathcal G^{i,n}_t\Big)=\mathcal F_t\quad\text{(up to $\mathbb P$--null sets)}.
\]
\end{assumption}

\begin{theorem}\label{thm:A-bar-to-Dirac}
Under Assumption~\ref{ass:all-improve}, for each $t\in[0,T]$,
\[
W_2\big(\tilde\pi^{\,n}_t,\delta_{S_t}\big)\ \xrightarrow[n]{\ L^1(\Omega)\ }\ 0.
\]
\end{theorem}

\begin{proof}
Fix $t\in[0,T]$. By the martingale convergence theorem,
\[
\mathbb E\big|S_t-\mathbb E[S_t\mid \mathcal G^{i,n}_t]\big|^2\ \xrightarrow[n]{}\ 0,\qquad i=1,\dots,m.
\]
By Lemma~\ref{lem:W2-identity}, $\mathbb E\,W_2^2(\pi^{i,n}_t,\delta_{S_t})=2\,\mathbb E|S_t-\mathbb E[S_t\mid \mathcal G^{i,n}_t]|^2\to 0$. Using Lemma~\ref{lem:bar-dom} with $\rho=\delta_{S_t}$,
\[
W_2^2\big(\tilde\pi^{\,n}_t,\delta_{S_t}\big)\ \le\ 4\,\max_{1\le i\le m} W_2^2\big(\pi^{i,n}_t,\delta_{S_t}\big),
\]
and hence, taking expectations and using $\max\le \sum$,
\[
\mathbb E\,W_2^2\big(\tilde\pi^{\,n}_t,\delta_{S_t}\big)\ \le\ 4\sum_{i=1}^m \mathbb E\,W_2^2\big(\pi^{i,n}_t,\delta_{S_t}\big)\ \xrightarrow[n]{}\ 0.
\]
By Cauchy--Schwarz, $\mathbb E\,W_2(\tilde\pi^{\,n}_t,\delta_{S_t})\le \sqrt{\mathbb E\,W_2^2(\tilde\pi^{\,n}_t,\delta_{S_t})}\to 0$.
\end{proof}

\begin{corollary}\label{cor:A-L2C}
If \eqref{eq:coeff-Lip} and \eqref{eq:compat} hold, then
\[
\mathbb E\Big[\sup_{t\le T}\big|\tilde S^{(n)}_t-S_t\big|^2\Big]\ \xrightarrow[n]{}\ 0.
\]
\end{corollary}

\begin{proof}
By Lemma~\ref{lem:bar-dom} with $\rho=\delta_{S_t}$,
\[
W_2^2\big(\tilde\pi^{\,n}_t,\delta_{S_t}\big)
\le 4\,\max_{1\le i\le m} W_2^2\big(\pi^{i,n}_t,\delta_{S_t}\big)
\le 4\sum_{i=1}^m W_2^2\big(\pi^{i,n}_t,\delta_{S_t}\big).
\]
Taking expectations and using Lemma~\ref{lem:W2-identity} together with the martingale convergence theorem (applied under Assumption~\ref{ass:all-improve}) gives, for each fixed $t$,
\[
\mathbb E\,W_2^2\big(\tilde\pi^{\,n}_t,\delta_{S_t}\big)
\le 4\sum_{i=1}^m \mathbb E\,W_2^2\big(\pi^{i,n}_t,\delta_{S_t}\big)\ \xrightarrow[n]{}\ 0.
\]
Moreover, still by Lemma~\ref{lem:W2-identity},
\[
\mathbb E\,W_2^2\big(\tilde\pi^{\,n}_t,\delta_{S_t}\big)
\ \le\ 4\sum_{i=1}^m \mathbb E\,W_2^2\big(\pi^{i,n}_t,\delta_{S_t}\big)
\ \le\ 8m\,\mathbb E|S_t|^2,
\]
and $t\mapsto \mathbb E|S_t|^2$ is integrable on $[0,T]$ by \eqref{eq:trueS}. Hence
\[
\int_0^T \mathbb E\,W_2^2\big(\tilde\pi^{\,n}_s,\delta_{S_s}\big)\,ds\ \xrightarrow[n]{}\ 0
\]
by dominated convergence. Inserting this into \eqref{eq:MV-stability} yields
$\mathbb E\big[\sup_{t\le T}|\tilde S^{(n)}_t-S_t|^2\big]\to 0$.
\end{proof}

\subsection{Failure of Convergence Under Non-Uniformly Increasing Information}

We first show that improvement of the combined information pool $\mathbb H^n$ alone is not sufficient to ensure convergence of the barycenter. 

\begin{assumption}\label{ass:join-only}
For each $t\in[0,T]$, $\mathcal H^{1}_t\subseteq \mathcal H^{2}_t\subseteq\cdots$ and $\sigma\!\big(\bigcup_{n\ge 1}\mathcal H^n_t\big)=\mathcal F_t$ ($\mathbb P$-a.s.), but for at least one expert $i$,
$\sigma\!\big(\bigcup_{n\ge 1}\mathcal G^{i,n}_t\big)\subsetneq \mathcal F_t$ ($\mathbb P$-a.s.).
\end{assumption}

\begin{example}\label{ex:failure-join}
Let $d=1$ and fix $t\in(0,T]$. On a product space supporting an $\mathbb F$--Brownian motion $W$, take independent Rademacher variables $U,V$ with $\mathbb P(U=\pm1)=\mathbb P(V=\pm1)=\tfrac12$, independent of $W$. Define $S_s:=UV$ for all $s\in[0,T]$, which solves \eqref{eq:trueS} with $\bar b\equiv 0=\bar\sigma$. Define
\[
\mathcal F_t:=\sigma(U,V)\ \vee\ \sigma(W_r:\,r\le t)\quad\text{(completed and right--continuous)},\qquad t\in[0,T].
\]
For all $n$ set
\[
\mathcal G^{1,n}_t:=\sigma(U)\ \vee\ \sigma(W_r:\,r\le t),\qquad
\mathcal G^{2,n}_t:=\sigma(V)\ \vee\ \sigma(W_r:\,r\le t),
\]
and, without changing notation, let $\mathbb G^{i,n}$ be their right--continuous, complete versions. Then
$\mathcal H^n_t=\sigma(U,V)\vee\sigma(W_r:\,r\le t)=\mathcal F_t$ for all $n$, so the join is already full, while
$\sigma(\bigcup_n\mathcal G^{1,n}_t)=\sigma(U)\vee\sigma(W_r:\,r\le t)\subsetneq\mathcal F_t$ and similarly for $i=2$.
Since $S_t=UV$ is independent of $W$ and $\mathbb P(V=\pm1\mid\sigma(U))=\tfrac12$, we have
\[
\pi^{1,n}_t=\pi^{2,n}_t=\tfrac12\delta_{-1}+\tfrac12\delta_{+1}\quad\text{a.s.,}\qquad
\delta_{S_t}\in\{\delta_{-1},\delta_{+1}\}\ \text{a.s.}
\]
Therefore, for any weights $w^{(n)}_{1,t},w^{(n)}_{2,t}\in\Delta_2$,
\[
\tilde\pi^{\,n}_t
=\mathrm{Bar}_{w^{(n)}_t}\Big(\tfrac12\delta_{-1}+\tfrac12\delta_{+1},\,\tfrac12\delta_{-1}+\tfrac12\delta_{+1}\Big)
=\tfrac12\delta_{-1}+\tfrac12\delta_{+1},
\]
and hence
\[
W_2\!\big(\tilde\pi^{\,n}_t,\delta_{S_t}\big)
=\begin{cases}
\sqrt{2}, & S_t=+1,\\
\sqrt{2}, & S_t=-1,
\end{cases}
\]
so $W_2(\tilde\pi^{\,n}_t,\delta_{S_t})=\sqrt{2}$ a.s.\ and in particular does not converge to $0$.
\end{example}

It is easy to show that this may cause the actual price process not to converge to the true price process. We do not pursue this here.

\subsection{Simulation of Convergence under Increasing Information}

We assume $d=1$ and a filtered space $(\Omega,\mathcal F,\mathbb F,\mathbb P)$ carrying a one--dimensional Brownian motion $W$. The \emph{true} price $S$ is a geometric Brownian motion (GBM)
\begin{equation}\label{eq:true-gbm}
dS_t=\bar b(t,S_t)\,dt+\bar\sigma(t,S_t)\,dW_t
\quad\text{with}\quad
\bar b(t,x)=\mu_\star x,\ \ \bar\sigma(t,x)=\sigma_\star x,\qquad S_0>0,
\end{equation}
where $\mu_\star$ and $\sigma_\star>0$ are some constants. Fix $m\in\mathbb N$ \textcolor{black}{traders} and information levels $n\in\mathbb N$. Let $X_t:=\log S_t$ and suppose \textcolor{black}{trader} $i$ at level $n$ observes
\[
Y^{i,n}_t \;=\; X_t + \varepsilon^{i,n}_t,\qquad \varepsilon^{i,n}_t\sim\mathcal N\!\big(0,\tau_i^2/n\big),
\]
independent of $W$ and across $(i,n)$. Then $\pi^{i,n}_t:=\mathcal L(S_t\mid Y^{i,n}_t)$ is lognormal. For weights $w^{(n)}\in\Delta_m$, define the (one--dimensional) $W_2$--barycenter of the traders at time $t$ by
\begin{equation}\label{eq:bar-def}
\tilde\pi^{\,n}_t \in \operatorname{argmin}_{\nu\in\Ptwo(\R_+)} 
\sum_{i=1}^m w^{(n)}_i\, W_2^2\!\big(\nu,\pi^{i,n}_t\big).
\end{equation}
In $d=1$, $\tilde\pi^{\,n}_t$ is the quantile average (comonotone coupling), which we evaluate in closed form via its first two moments.

\[
\tilde\pi^{\,n}_t=\mathcal L\!\Big(\sum_{i=1}^m w^{(n)}_{i,t}\,\mathrm{e}^{\,m^{i,n}_t+s^{i,n}_t Z}\Big),\quad Z\sim\mathcal N(0,1),\ \ m^{i,n}_t:=\mathbb E[\log S_t\mid\mathcal G^{i,n}_t],\ \ (s^{i,n}_t)^2:=\Var(\log S_t\mid\mathcal G^{i,n}_t).
\]

We assume the drift and volatility depend on the barycentric mean and standard deviation:
\[
m_1(\mu):=\int y\,\mu(dy),\qquad s(\mu):=\sqrt{\int (y-m_1(\mu))^2\,\mu(dy)}.
\]
For $\kappa_d,\kappa_v\ge 0$, set for $(t,x,\mu)\in[0,T]\times\R_+\times\Ptwo(\R_+)$
\begin{equation}\label{eq:coeffs}
b(t,x,\mu)\;:=\; x\!\left(\mu_\star + \kappa_d \log\!\frac{m_1(\mu)}{x}\right),\qquad
\sigma(t,x,\mu)\;:=\; x\,\sigma_\star\!\left(1+\kappa_v\,\frac{s(\mu)}{m_1(\mu)}\right).
\end{equation}
The \textcolor{black}{actual market} price $\tilde S^{(n)}$ follows the McKean-Vlasov-type SDE driven by the same $W$:
\begin{equation}\label{eq:synthetic}
d\tilde S^{(n)}_t \;=\; b\!\big(t,\tilde S^{(n)}_t,\tilde\pi^{\,n}_t\big)\,dt
\;+\; \sigma\!\big(t,\tilde S^{(n)}_t,\tilde\pi^{\,n}_t\big)\,dW_t,
\qquad \tilde S^{(n)}_0=S_0.
\end{equation}
By construction,
\begin{equation}\label{eq:compat}
b(t,x,\delta_x)=\mu_\star x=\bar b(t,x),\qquad 
\sigma(t,x,\delta_x)=\sigma_\star x=\bar\sigma(t,x),
\end{equation}
so \eqref{eq:synthetic} is \emph{compatible} with \eqref{eq:true-gbm} in the sense $b(\cdot,\cdot,\delta_x)=\bar b$, $\sigma(\cdot,\cdot,\delta_x)=\bar\sigma$. As $n\to\infty$ the posterior variances scale as $\tau_i^2/n$, the barycenter $\tilde\pi^{\,n}_t\Rightarrow\delta_{S_t}$, and under standard Lipschitz and linear growth conditions of \eqref{eq:coeffs} one has convergence of $\tilde S^{(n)}$ to $S$ in $L^2(\Omega;C)$.
We simulate $30$ paths on $[0,1]$ year with daily steps. For each $n\in\{1,10,100,1000\}$ we draw a common Brownian path per row and plot: left (blue) the true $S$, right (red) the synthetic $\tilde S^{(n)}$ built from the $W_2$--barycenter \eqref{eq:bar-def}. Parameters used in the figure: 
$S_0=100$, $\mu_\star=8\%$, $\sigma_\star=60\%$, $m=4$, $w=(0.4,0.3,0.2,0.1)$, $\tau=(2.0,1.2,2.5,1.5)$, $\kappa_d=0.35$, $\kappa_v=2.75$.


\medskip

It must be noted that other choices are possible, as long as they satisfy the compatibility condition. Some suggestions are given below. For $\mu\in\mathcal P_2(\mathbb R_+)$ set
$m_1(\mu):=\int y\,\mu(dy)$,
$s(\mu):=\sqrt{\int (y-m_1(\mu))^2\,\mu(dy)}$,
$\mathrm{cv}(\mu):= s(\mu)/(m_1(\mu)+\varepsilon)$ for small $\varepsilon>0$.
Let $\bar b,\bar\sigma$ be the true coefficients from \eqref{eq:trueS}.
All parameters $\kappa,\kappa_d,\kappa_v,\lambda\ge 0$.

\begin{table}[htbp]
\centering
\footnotesize
\begin{tabularx}{\linewidth}{>{\raggedright\arraybackslash}p{0.8em} >{\raggedright\arraybackslash}X >{\raggedright\arraybackslash}X}
\toprule
\# & Drift $b(t,x,\mu)$ & Volatility $\sigma(t,x,\mu)$ \\
\midrule
1 &
$\displaystyle \bar b(t,x)\;+\;\kappa\big(m_1(\mu)-x\big)$ &
$\displaystyle \bar\sigma(t,x)$ \\[0.35em]
2 &
$\displaystyle \bar b(t,x)$ &
$\displaystyle \bar\sigma(t,x)\!\left(1+\kappa\,\mathrm{cv}(\mu)\right)$ \\[0.35em]
3 &
$\displaystyle \bar b(t,x)\;+\;\kappa_d\big(m_1(\mu)-x\big)$ &
$\displaystyle \bar\sigma(t,x)\!\left(1+\kappa_v\,\mathrm{cv}(\mu)\right)$ \\[0.35em]
4 &
$\displaystyle \bar b(t,x)\Big(1+\kappa\!\left(\tfrac{m_1(\mu)}{x}-1\right)\Big)$ &
$\displaystyle \bar\sigma(t,x)$ \\[0.35em]
5 &
$\displaystyle \bar b(t,x)\;-\;\lambda\,x\,\mathrm{cv}(\mu)^2$ &
$\displaystyle \bar\sigma(t,x)$ \\[0.35em]
6 &
$\displaystyle \bar b(t,x)$ &
$\displaystyle \bar\sigma(t,x)\,\sqrt{\,1+\kappa\!\left(\tfrac{m_1(\mu)-x}{m_1(\mu)+\varepsilon}\right)^{\!2}}$ \\
\bottomrule
\end{tabularx}
\caption{Each pair satisfies the compatibility condition: for $\mu=\delta_x$ one has $m_1(\delta_x)=x$, $s(\delta_x)=\mathrm{cv}(\delta_x)=0$, hence $b(t,x,\delta_x)=\bar b(t,x)$ and $\sigma(t,x,\delta_x)=\bar\sigma(t,x)$.}
\label{tab:compatible-simple}
\end{table}

However, it is obvious that any choice of the drift and volatility must be meaningful from the modeling perspective, i.e. admit financial interpretation. In the next sections, we will provide more concrete (and much more sophisticated) structures that admit clear interpretation and generate a wide range of nontrivial implications and results.

\begin{figure}[p]  
    \centering
    \includegraphics[width=\linewidth,height=0.9\textheight,keepaspectratio]{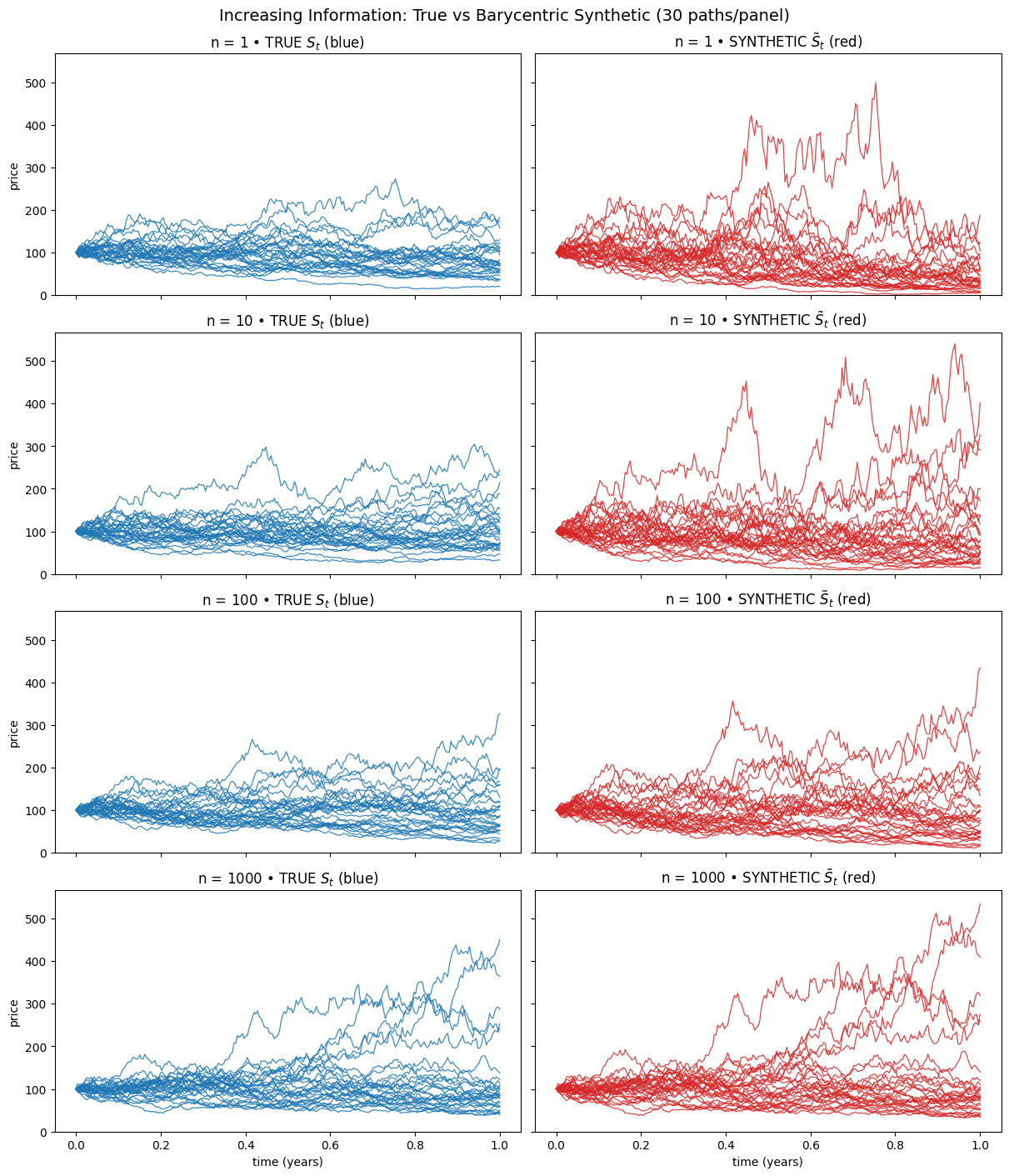}
    \caption{\textbf{Increasing information.} Rows $n=1,10,100,1000$; left: true price $S$ (blue); right: synthetic $\tilde S^{(n)}$ (red). Common Brownian shocks per row; we observe convergence as the dispersion in traders' posteriors vanishes.}
    \label{fig:Sim1}
\end{figure}

\clearpage

\section{\textcolor{black}{Individual} Biases under Increasing Information}\label{sec:info-drift}

\subsection{Preliminaries and Model Set-Up}\label{subsec:setup}

In this section, we propose a much more explicit and financially motivated drift perturbation scheme capturing model uncertainty in which \textcolor{black}{a single trader's} belief is represented by \emph{a drift perturbation term $\rho$} which is convexly combined with the true drift $\alpha$,
with a (random, time–varying) \emph{bias weight $\beta$ that shrinks to $0$ as information increases}. \textcolor{black}{Thus, we want to study how a single trader uses the information available to them, combining observations with personal biases to arrive at a candidate price process. In the next section we will study how a trader seeking arbitrage opportunities should aggregate the individual (biased) proposal processes.} This is an elaboration on the ideas of the previous section, showing more explicitly how an increasing information flow implies convergence of the proposal process to the true price process \textcolor{black}{for a single trader}. In this setting, the weight of the bias depends on the distance between the conditional distribution of the price $S_t$ given the available information $\mathcal G^n_t$ and the value of the optimal estimate of $S_t$ provided by the filtered $\widehat S_t$. This structure corresponds to the intuition that the \emph{the distance between the conditional distribution and the estimated value captures ambiguity} regarding the true value of the price process, and this ambiguity shrinks as more information is revealed, as $n\to \infty$, eventually collapsing to a Dirac measure at the correct value. \textcolor{black}{Thus, this more specialized model also incorporates a fairly novel way of measuring the impact of a trader's intuitive sense of ambiguity regarding the true value of a partially observed price process.}

Let $(\Omega,\mathcal F, \mathbb F, \mathbb P)$ be a complete probability space
with a right–continuous, complete filtration
$\mathbb F=(\mathcal F_t)_{t\in[0,T]}$ supporting a $D$–dimensional Brownian motion $W$.
The true $d$–dimensional price process $S$ is the unique strong solution of
\begin{equation}\label{eq:trueS-drift}
dS_t=\alpha(t,S_t)\,dt+\sigma(t,S_t)\,dW_t,\qquad S_0\in L^2(\Omega;\mathbb R^d),
\end{equation}
where the coefficients $\alpha:[0,T]\times\mathbb R^d\to\mathbb R^d$ and
$\sigma:[0,T]\times\mathbb R^d\to\mathbb R^{d\times D}$ satisfy:

\medskip
\begin{assumption}\label{ass:LS}
There exists $L\ge 1$ such that for all $t\in[0,T]$, $x,y\in\mathbb R^d$,
\[
\begin{aligned}
&|\alpha(t,x)-\alpha(t,y)|+\|\sigma(t,x)-\sigma(t,y)\|\le L|x-y|,\\
&|\alpha(t,x)|^2+\|\sigma(t,x)\|^2\le L\,(1+|x|^2).
\end{aligned}
\]
\end{assumption}
Under Assumption~\ref{ass:LS}, \eqref{eq:trueS-drift} is well posed and
$\mathbb E\!\left[\sup_{t\le T}|S_t|^2\right]<\infty$.

\medskip
We model increasingly informative observers by a fixed index $n\in\mathbb N$
and a subfiltration $\mathbb G^n=(\mathcal G^n_t)_{t\in[0,T]}$ of $\mathbb F$ such that for each $t\in[0,T]$
\begin{equation}\label{eq:Y-increasing}
\mathcal G^1_t\subseteq \mathcal G^2_t\subseteq\cdots,\qquad
\sigma\Big(\bigcup_{n\ge 1}\mathcal G^n_t\Big)=\mathcal F_t\quad(\text{up to }\mathbb P\text{-null sets}).
\end{equation}
We assume right–continuity and completeness of each $\mathbb G^n$ without changing notation.

Fix $n$ and $t\in[0,T]$. Set
\[
\widehat S^{(n)}_t:=\mathbb E[S_t\mid \mathcal G^n_t],\qquad
\pi^{n}_t:=\mathcal{L}(S_t\mid \mathcal G^n_t)\in\mathcal P_2(\mathbb R^d).
\]
Recall the identity (see Lemma~\ref{lem:W2-identity})
\begin{equation}\label{eq:W2-is-Var}
W_2^2\big(\pi^{n}_t,\delta_{\widehat S^{(n)}_t}\big)
=\mathbb E\!\left[\,|S_t-\widehat S^{(n)}_t|^2\ \Big|\ \mathcal G^n_t\right]
=\mathrm{Var}\big(S_t\mid \mathcal G^n_t\big).
\end{equation}
Define the \emph{measure of ambiguity} and the \emph{bias weight} correspondingly by
\begin{equation}\label{eq:gamma-beta-def}
\gamma^{(n)}_t:=\sqrt{\mathrm{Var}(S_t\mid \mathcal G^n_t)}
= W_2\!\big(\pi^{n}_t,\delta_{\widehat S^{(n)}_t}\big),\qquad
\beta^{(n)}_t:=\beta\!\left(\gamma^{(n)}_t\right),
\end{equation}
where $\beta:[0,\infty)\to[0,1]$ satisfies:

\medskip
\begin{assumption}\label{ass:beta}
$\beta$ is continuous at $0$ with $\beta(0)=0$, bounded by $1$, and locally Lipschitz on $[0,\infty)$.
\end{assumption}

Define the \emph{bias-perturbed drift} by
\begin{equation}\label{eq:convex-drift}
\alpha^{(n)}_\beta(t,x,\omega)
:=(1-\beta^{(n)}_t(\omega))\,\alpha(t,x)+\beta^{(n)}_t(\omega)\,\rho^{(n)}_t(\omega).
\end{equation}
Consider \emph{the proposed synthetic price process} $\widetilde S^{(n)}$ given by:
\begin{equation}\label{eq:tildeS-beta}
d\widetilde S^{(n)}_t=\alpha^{(n)}_\beta\big(t,\widetilde S^{(n)}_t\big)\,dt
+\sigma\big(t,\widetilde S^{(n)}_t\big)\,dW_t,\qquad \widetilde S^{(n)}_0=S_0.
\end{equation}

\noindent Note $\gamma^{(n)}_t$ is $\mathcal G^n_t$–measurable by \eqref{eq:W2-is-Var}. By considering the $\mathbb G^n$-progressively measurable versions of $\pi_t^n$ (see the remark \ref{progressiveversions} in section 4) and the optional modifications of $\widehat S^{(n)}_t$, we obtain that  $\beta^{(n)}$ can be chosen to be $\mathbb G^n$–progressively measurable.

\begin{lemma}\label{lem:beta-vanish}
Assume \eqref{eq:Y-increasing}. Then for each $t\in[0,T]$,
\[
\mathbb E\big[(\gamma^{(n)}_t)^2\big]=\mathbb E\big|S_t-\widehat S^{(n)}_t\big|^2\ \xrightarrow[n]{}\ 0,
\qquad
\mathbb E\big[(\beta^{(n)}_t)^2\big]\ \xrightarrow[n]{}\ 0,
\]
and by dominated convergence,
\begin{equation}\label{eq:int-beta2-to-0}
\int_0^T \mathbb E\big[(\beta^{(n)}_t)^2\big]\,dt\ \xrightarrow[n]{}\ 0.
\end{equation}
\end{lemma}

\begin{proof}
By \eqref{eq:Y-increasing} and the martingale convergence theorem,
$\widehat S^{(n)}_t\to S_t$ in $L^2$, hence $\mathbb E(\gamma^{(n)}_t)^2=\mathbb E|S_t-\widehat S^{(n)}_t|^2\to 0$.
By Assumption~\ref{ass:beta}, $\beta$ is locally Lipschitz at $0$ and $\beta(0)=0$, hence
$(\beta^{(n)}_t)^2\le C\,(\gamma^{(n)}_t)^2$ for all $\gamma^{(n)}_t$ sufficiently small. 
Set $\gamma_n:=\gamma_t^{(n)}$ and $\beta_n:=\beta(\gamma_n)$; fix $\delta>0$ and let $L_\delta$ be the Lipschitz constant of $\beta$ on $[-\delta,\delta]$. Then, using $0\le(\beta^{(n)}_t)^2\le 1$ and Chebyshev's inequality, we have
\[
\mathbb{E}[\beta_n^2]
=\mathbb{E}\big[\beta_n^2 \mathbf{1}_{\{|\gamma_n|\le\delta\}}\big]
+\mathbb{E}\big[\beta_n^2 \mathbf{1}_{\{|\gamma_n|>\delta\}}\big]
\le L_\delta^2\,\mathbb{E}\big[\gamma_n^2 \mathbf{1}_{\{|\gamma_n|\le\delta\}}\big]
+\mathbb{P}(|\gamma_n|>\delta)
\le \big(L_\delta^2+\delta^{-2}\big)\,\mathbb{E}[\gamma_n^2]
\xrightarrow[n\to\infty]{}0.
\]
Hence \eqref{eq:int-beta2-to-0} follows by dominated convergence.
\end{proof}

Let $\rho^{(n)}=(\rho^{(n)}_t)_{t\in[0,T]}$ be a $\mathbb G^n$–progressively measurable $\mathbb R^d$–valued process interpreted as the \textcolor{black}{trader's} opinion on the correct drift. We assume:

\begin{assumption}\label{ass:rho}
There exists $p>1$ such that
\[
\sup_{n\ge 1}\ \int_0^T \mathbb E\big[\,|\rho^{(n)}_t|^{2p}\,\big]\,dt\ <\ \infty.
\]
\end{assumption}

\begin{proposition}\label{prop:wellposed}
Under Assumptions~\ref{ass:LS},~\ref{ass:beta},~\ref{ass:rho}, for every $n$ the SDE \eqref{eq:tildeS-beta} admits a unique strong solution
with
\[
\mathbb E\Big[\sup_{t\le T}\big|\widetilde S^{(n)}_t\big|^2\Big]\ \le\
C\Big(1+\mathbb E|S_0|^2+\mathbb E\int_0^T |\rho^{(n)}_t|^2\,dt\Big),
\]
for a constant $C=C(L,T)$ independent of $n$.
\end{proposition}

\begin{proof}
The drift $x\mapsto \alpha^{(n)}_\beta(t,x)$ is globally Lipschitz with the same constant $L$
as $\alpha$, since $\beta^{(n)}_t\in[0,1]$ and $\rho^{(n)}_t$ does not depend on $x$.
Moreover,
\[
|\alpha^{(n)}_\beta(t,x)|\le |\alpha(t,x)|+|\rho^{(n)}_t|
\le C(1+|x|)+|\rho^{(n)}_t|
\]
with $C=C(L)$. Standard SDE estimates (e.g.\ Itô, BDG, Gronwall) yield the moment bound under Assumption~\ref{ass:rho}.
\end{proof}

\subsection{Stability and Convergence to the True Process}\label{subsec:stability}

Set $\Delta^{(n)}_t:=\widetilde S^{(n)}_t-S_t$. Using \eqref{eq:trueS-drift} and \eqref{eq:tildeS-beta},
\begin{equation}\label{eq:Delta-eq}
d\Delta^{(n)}_t=\Big(\alpha\big(t,\widetilde S^{(n)}_t\big)-\alpha(t,S_t)\Big)\,dt
+\Big(\sigma\big(t,\widetilde S^{(n)}_t\big)-\sigma(t,S_t)\Big)\,dW_t
+\beta^{(n)}_t\Big(\rho^{(n)}_t-\alpha\big(t,\widetilde S^{(n)}_t\big)\Big)\,dt.
\end{equation}

\begin{theorem}\label{thm:stab}
Under Assumptions~\ref{ass:LS}, ~\ref{ass:beta},~\ref{ass:rho}, there exists $C=C(L,T)$ such that for every $n$ and $t\in[0,T]$,
\begin{equation}\label{eq:stab-est}
\mathbb E\Big[\sup_{s\le t}\big|\Delta^{(n)}_s\big|^2\Big]
\ \le\ C\int_0^t \mathbb E\Big[\big|\Delta^{(n)}_u\big|^2+(\beta^{(n)}_u)^2\big(1+|\rho^{(n)}_u|^2+|S_u|^2\big)\Big]\,du.
\end{equation}
\end{theorem}

\begin{proof}
Applying BDG and the Lipschitz properties of $(\alpha,\sigma)$,
\[
\mathbb E\!\left[\sup_{s\le t}\big|\Delta^{(n)}_s\big|^2\right]
\le C\int_0^t \mathbb E\big|\Delta^{(n)}_u\big|^2\,du
+ C\int_0^t \mathbb E\!\left[(\beta^{(n)}_u)^2\,\big|\rho^{(n)}_u-\alpha(u,\widetilde S^{(n)}_u)\big|^2\right]\,du.
\]
Using $(a+b)^2\le 2(a^2+b^2)$ and the linear growth of $\alpha$,
\[
\big|\rho^{(n)}_u-\alpha(u,\widetilde S^{(n)}_u)\big|^2
\le 2|\rho^{(n)}_u|^2+2\,|\alpha(u,\widetilde S^{(n)}_u)|^2
\le C\big(1+|\widetilde S^{(n)}_u|^2+|\rho^{(n)}_u|^2\big).
\]
By Proposition~\ref{prop:wellposed} and $(a+b)^2\le 2(a^2+b^2)$,
$|\widetilde S^{(n)}_u|^2\le 2|S_u|^2+2|\Delta^{(n)}_u|^2$, from which we obtain
\[
\big|\rho^{(n)}_u-\alpha(u,\widetilde S^{(n)}_u)\big|^2
\le C\big(1+|S_u|^2+|\Delta^{(n)}_u|^2+|\rho^{(n)}_u|^2\big).
\]
Inserting this bound and absorbing the resulting $(\beta^{(n)}_u)^2|\Delta^{(n)}_u|^2$ term
into the first integral (since $(\beta^{(n)}_u)^2\le 1$), we obtain \eqref{eq:stab-est}.
\end{proof}

\begin{theorem}\label{thm:conv}
Assume ~\ref{ass:LS}, ~\ref{ass:beta}, ~\ref{ass:rho} and \eqref{eq:Y-increasing}. Then
\[
\mathbb E\Big[\sup_{t\le T}\big|\widetilde S^{(n)}_t-S_t\big|^2\Big]\ \xrightarrow[n]{}\ 0.
\]
\end{theorem}

\begin{proof}
From \eqref{eq:stab-est} with $t=T$ and Gronwall's inequality,
\[
\mathbb E\Big[\sup_{s\le T}\big|\Delta^{(n)}_s\big|^2\Big]
\ \le\ C\int_0^T \mathbb E\Big[(\beta^{(n)}_u)^2\big(1+|\rho^{(n)}_u|^2+|S_u|^2\big)\Big]\,du,
\]
for $C=C(L,T)$. Write
\[
A_n:=\int_0^T \mathbb E\big[(\beta^{(n)}_u)^2\big]\,du,\qquad
B_n:=\int_0^T \mathbb E\big[(\beta^{(n)}_u)^2|\rho^{(n)}_u|^2\big]\,du,\qquad
C_n:=\int_0^T \mathbb E\big[(\beta^{(n)}_u)^2|S_u|^2\big]\,du.
\]
By Lemma~\ref{lem:beta-vanish}, $A_n\to 0$.

For $B_n$, apply Hölder's inequality on the product space $([0,T]\times\Omega,\,dt\otimes d\mathbb P)$ with conjugate exponents $p>1$, $q:=\frac{p}{p-1}\in(1,\infty)$:
\[
B_n=\|(\beta^{(n)})^2\,|\rho^{(n)}|^2\|_{L^1}
\ \le\ \|(\beta^{(n)})^2\|_{L^q}\ \||\rho^{(n)}|^2\|_{L^p}.
\]
Since $0\le (\beta^{(n)})^2\le 1$, we have $(\beta^{(n)})^{2q}\le (\beta^{(n)})^2$ and hence
$\|(\beta^{(n)})^2\|_{L^q}\le A_n^{1/q}$. By Assumption~\ref{ass:rho},
$\||\rho^{(n)}|^2\|_{L^p}=(\int_0^T\mathbb E|\rho^{(n)}_u|^{2p}du)^{1/p}\le C$ uniformly in $n$.
Therefore $B_n\le C\,A_n^{1/q}\to 0$.

For $C_n$, fix $A>0$ and split
\[
C_n
=\int_0^T \mathbb E\big[(\beta^{(n)}_u)^2|S_u|^2\mathbf 1_{\{|S_u|\le A\}}\big]\,du
+\int_0^T \mathbb E\big[(\beta^{(n)}_u)^2|S_u|^2\mathbf 1_{\{|S_u|>A\}}\big]\,du
\ \le\ A^2 A_n+\int_0^T \mathbb E\big[|S_u|^2\mathbf 1_{\{|S_u|>A\}}\big]\,du.
\]
The second term is independent of $n$ and goes to $0$ as $A\to\infty$ because
$\int_0^T \mathbb E|S_u|^2\,du<\infty$ (from \eqref{eq:trueS-drift} and Assumption~\ref{ass:LS}).
Hence, given $\varepsilon>0$, choose $A$ such that the tail term $\le\varepsilon$, and then $n$ large so that $A^2 A_n\le\varepsilon$, yielding $C_n\le 2\varepsilon$ for $n$ large. Combining, we obtain $A_n+B_n+C_n\to 0$, and the claim follows.
\end{proof}

The above analysis hinges on the following precise identification, which also
motivates our choice of the functional form for measuring ambiguity.

\begin{proposition}\label{prop:cv-is-w2}
For any sub–$\sigma$–algebra $\mathcal G\subset\mathcal F$ and any $X\in L^2(\Omega;\mathbb R^d)$,
\[
\inf_{z\in\mathbb R^d} W_2^2\big(\mathcal{L}(X\mid\mathcal G),\delta_z\big)
= W_2^2\big(\mathcal{L}(X\mid\mathcal G),\delta_{\mathbb E[X\mid\mathcal G]}\big)
=\mathrm{Var}(X\mid\mathcal G).
\]
Consequently, the choice
\[
\beta^{(n)}_t=\beta\!\Big(W_2\!\big(\mathcal{L}(S_t\mid\mathcal G^n_t),\delta_{\widehat S^{(n)}_t}\big)\Big)
=\beta\!\Big(\sqrt{\mathrm{Var}(S_t\mid\mathcal G^n_t)}\Big)
\]
shrinks to $0$ exactly at the (square–root) rate at which the conditional variance vanishes.
\end{proposition}

\begin{proof}
Immediate from Lemma~\ref{lem:W2-identity} by minimizing over $z$.
\end{proof}

\subsection{ Rate of Information Convergence}\label{subsec:rates}

Suppose, in addition, that for some $\eta\in(0,1]$ and all $u\in[0,T]$,
\begin{equation}\label{eq:var-rate}
\mathbb E\,\mathrm{Var}\!\big(S_u\mid \mathcal G^n_u\big)\ \le\ C\,n^{-2\eta}
\qquad\text{(a numerical rate of information gain).}
\end{equation}
Assume further that $\beta$ is \emph{globally} Lipschitz on $[0,\infty)$ with constant $L_\beta$ and $\beta(0)=0$.
Then, recalling $\gamma^{(n)}_u=\sqrt{\mathrm{Var}(S_u\mid\mathcal G^n_u)}$,
\[
\big(\beta^{(n)}_u\big)^2=\beta(\gamma^{(n)}_u)^2\ \le\ L_\beta^2\,\mathrm{Var}\!\big(S_u\mid\mathcal G^n_u\big),
\qquad\text{hence}\qquad
\int_0^T \mathbb E\big(\beta^{(n)}_u\big)^2\,du\ \le\ C\,L_\beta^2\,n^{-2\eta}.
\]

To convert this into a rate for $\mathbb E\big[\sup_{t\le T}|\widetilde S^{(n)}_t-S_t|^2\big]$, we need to control the
mixed terms containing $|S|^2$ and $|\rho^{(n)}|^2$ appearing in \eqref{eq:stab-est}.
Assume the following fourth--moment bounds:
\begin{equation}\label{eq:fourth-moments}
\int_0^T \mathbb E|S_u|^4\,du\ <\ \infty,
\qquad
\sup_{n\ge 1}\ \int_0^T \mathbb E\big[|\rho^{(n)}_u|^4\big]\,du\ <\ \infty.
\end{equation}

\begin{proposition}\label{prop:rate}
Under \eqref{eq:var-rate}, the global Lipschitz condition on $\beta$, and \eqref{eq:fourth-moments},
there exists $C'=C'(L,T,L_\beta,\text{data})$ such that
\[
\mathbb E\Big[\sup_{t\le T}\big|\widetilde S^{(n)}_t-S_t\big|^2\Big]\ \le\ C'\,n^{-\eta},\qquad n\to\infty.
\]
\end{proposition}

\begin{proof}
From \eqref{eq:stab-est} with $t=T$ and Gronwall's inequality,
\[
\mathbb E\Big[\sup_{s\le T}\big|\Delta^{(n)}_s\big|^2\Big]
\ \le\ C\!\int_0^T \mathbb E\Big[(\beta^{(n)}_u)^2\Big]\,du
+ C\!\int_0^T \mathbb E\Big[(\beta^{(n)}_u)^2|\rho^{(n)}_u|^2\Big]\,du
+ C\!\int_0^T \mathbb E\Big[(\beta^{(n)}_u)^2|S_u|^2\Big]\,du,
\]
with $C=C(L,T)$. Denote
\[
A_n:=\int_0^T \mathbb E\big[(\beta^{(n)}_u)^2\big]\,du,\qquad
B_n:=\int_0^T \mathbb E\big[(\beta^{(n)}_u)^2|\rho^{(n)}_u|^2\big]\,du,\qquad
C_n:=\int_0^T \mathbb E\big[(\beta^{(n)}_u)^2|S_u|^2\big]\,du.
\]
By the computation above, $A_n\le C L_\beta^2 n^{-2\eta}$. Since $0\le \beta^{(n)}_u\le 1$, we have $(\beta^{(n)}_u)^4\le (\beta^{(n)}_u)^2$ and hence
\[
\int_0^T \mathbb E\big[(\beta^{(n)}_u)^4\big]\,du\ \le\ A_n\ \le\ C L_\beta^2 n^{-2\eta}.
\]
By Hölder's inequality on $([0,T]\times\Omega,dt\otimes d\mathbb P)$ with exponents $(2,2)$ and \eqref{eq:fourth-moments},
\[
B_n\le \Big(\int_0^T \mathbb E[(\beta^{(n)}_u)^4]\,du\Big)^{1/2}
      \Big(\int_0^T \mathbb E[|\rho^{(n)}_u|^4]\,du\Big)^{1/2}
   \ \le\ C\,n^{-\eta},
\]
and similarly
\[
C_n\le \Big(\int_0^T \mathbb E[(\beta^{(n)}_u)^4]\,du\Big)^{1/2}
      \Big(\int_0^T \mathbb E[|S_u|^4]\,du\Big)^{1/2}
   \ \le\ C\,n^{-\eta}.
\]
Thus $A_n+B_n+C_n\le C'(n^{-2\eta}+n^{-\eta}+n^{-\eta})\le C''\,n^{-\eta}$, and the claim follows.
\end{proof}

\subsection{Simulation of the Shrinking Bias}

We work on a filtered probability space $(\Omega,\mathcal F,\mathbb F,\mathbb P)$ carrying a one–dimensional Brownian motion $W$. The \emph{true} price process $S$ solves
\begin{equation}\label{eq:true}
dS_t=\alpha(t,S_t)\,dt+\sigma(t,S_t)\,dW_t,\qquad S_0>0,
\end{equation}
and for the simulation we again take the geometric Brownian motion (GBM) model
\[
\alpha(t,x)=\mu_\star x,\qquad \sigma(t,x)=\sigma_\star x.
\]

We set the weight function to be  
\[
\beta^{(n)}_t \;=\; 1-\exp\!\Big(-\kappa_b\,(\gamma^{(n)}_t)^{p_b}\Big) \;=\; 1-\exp\!\Big(-\kappa_b\,\big(\mathrm{Var}(S_t\mid \mathcal G^n_t)\big)^{p_b/2}\Big),\quad \kappa_b=10^{-3},\ p_b=2.4.
\]

We use a specification of $\rho^{(n)}$ proportional to the conditional mean (amplified by relative ambiguity):
\[
\rho^{(n)}_t=\mu_{\mathrm{op}}\,\widehat S^{(n)}_t\Big(1+c_{\mathrm{rel}}\frac{\gamma^{(n)}_t}{\widehat S^{(n)}_t+\varepsilon}\Big),\qquad
\mu_{\mathrm{op}}>0,\ c_{\mathrm{rel}}\ge 0,\ \varepsilon>0.
\]

The synthetic price $\widetilde S^{(n)}$ is obtained by convexly mixing the true drift with the opinion drift, with $\beta^{(n)}$ as weight, and keeping the same diffusion and the same Brownian motion $W$:
\begin{equation}\label{eq:synthetic}
d\widetilde S^{(n)}_t=\Big[(1-\beta^{(n)}_t)\,\alpha\!\big(t,\widetilde S^{(n)}_t\big)
+\beta^{(n)}_t\,\rho^{(n)}_t\Big]\,dt
+\sigma\!\big(t,\widetilde S^{(n)}_t\big)\,dW_t,\qquad \widetilde S^{(n)}_0=S_0.
\end{equation}
At zero ambiguity ($\gamma^{(n)}_t=0$) one has $\beta^{(n)}_t=0$, so \eqref{eq:synthetic} reduces to the true dynamics \eqref{eq:true}. Under the standard Lipschitz and linear–growth assumptions, increasing information implies $\gamma^{(n)}\to 0$, which in turn implies $\beta^{(n)}\to 0$, and stability yields $\widetilde S^{(n)}\to S$ in $L^2(\Omega;C)$. The common Brownian motion $W$ couples the paths so differences are solely driven by the bias term.
  
   Below are $4\times2$ plots, with rows corresponding to information levels $n\in\{1,10,100,1000\}$. The true process $S$ (blue) is on the left, the synthetic process $\widetilde S^{(n)}$ (red) is on the right. We generated $30$ paths per panel with daily steps over $[0,1]$ year. 

\begin{figure}[p]
  \centering
 \includegraphics[width=\linewidth,height=0.9\textheight,keepaspectratio]{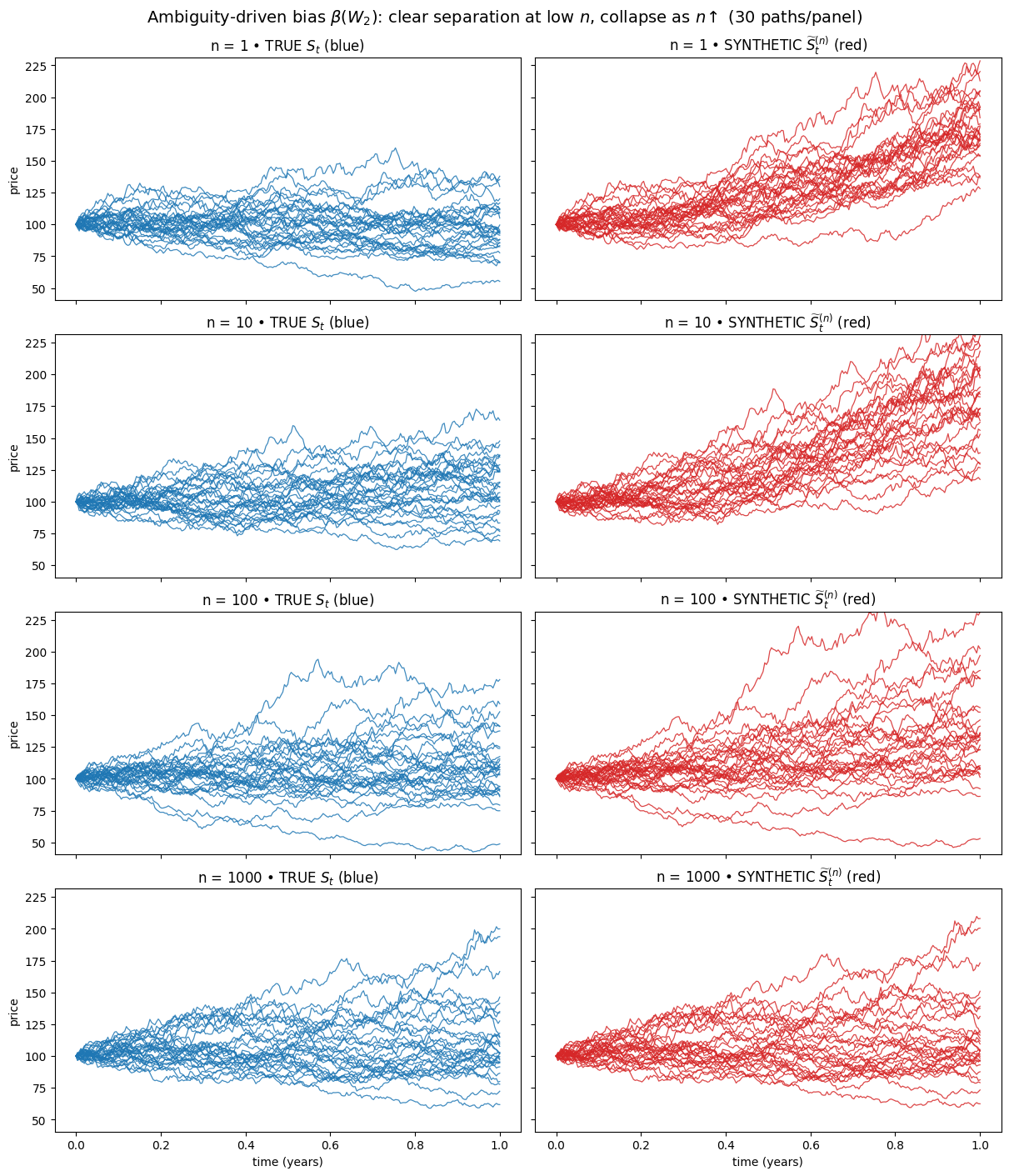}
  \caption{Ambiguity–driven bias: at low information ($n=1,10$) the synthetic process deviates due to a nonzero $\beta^{(n)}$, while as $n$ increases ($100,1000$) the bias weight shrinks and $\widetilde S^{(n)}$ collapses to $S$.}
  \label{fig:Sim2}
\end{figure}

\clearpage

\subsection{Conclusions}

We have constructed a rigorous, explicit mechanism in which a \textcolor{black}{trader's} opinion drift
$\rho^{(n)}$ is convexly combined with the true drift $\alpha$, with weight
\[
\beta^{(n)}_t=\beta\!\Big(W_2\!\big(\mathcal{L}(S_t\mid\mathcal G^n_t),\delta_{\widehat S^{(n)}_t}\big)\Big)
=\beta\!\Big(\sqrt{\mathrm{Var}(S_t\mid\mathcal G^n_t)}\Big)\ \in[0,1].
\]
\emph{The Wasserstein distance measures a sense of ambiguity regarding using the estimate $\widehat S_t$ instead of the true value of $S_t$}. As the sub–$\sigma$–algebras $\mathcal G^n_t$ increase to $\mathcal F_t$ (for each $t$),
$\beta^{(n)}$ vanishes in $L^2(dt\otimes d\mathbb P)$, and the modified (synthetic) price process
\eqref{eq:tildeS-beta} converges to the true price process \eqref{eq:trueS-drift} in
$L^2(\Omega;C([0,T];\mathbb R^d))$ (\Cref{thm:conv}). The shrinking mechanism is
explicitly tied to the conditional variance of $S_t$ given the current information level (Proposition~\ref{prop:cv-is-w2}).

\section{Optimal Aggregation of Expert Opinions under Information Constraints}
\label{sec:KL-budget-aggregation}

\subsection{Preliminaries and Model Set-Up}
\label{subsec:preliminaries-KL}

We now study an optimal aggregation problem where a \emph{trader is acting under information constraints $\mathbb G \subset \mathbb F$ and consults a (now possibly continuous) set of experts}, represented by some compact space $\Lambda$, who propose corrections to the observed drift, represented by expert-indexed flows of probability measures $\rho^\lambda$, which the \emph{trader wishes to aggregate to obtain the best estimate of the unknown true drift $a$}. The trader tries to minimize the distance between the aggregated correction term $\bar\rho$ and the trader's own estimate $\widehat a$ based on available information $\mathbb G$, taking into account their own prior beliefs on the expert community, represented by a flow $\pi$ of probability measures over $\Lambda$. \textcolor{black}{From a financial point of view, the trader is seeking positive alphas as defined in \cite{JarrowProtter}, i.e. an arbitrage opportunity or a dominated asset.}

We fix a finite horizon $T>0$ and a filtered probability space
$(\Omega,\mathcal F,\mathbb F=(\mathcal F_t)_{0\le t\le T},\mathbb P)$ satisfying the usual conditions. Let $W$ and $B$ be independent one–dimensional $\mathbb F$–Brownian motions.
Assume the stock price follows
\begin{equation}\label{eq:signal-SDE}
dS_t = a_t\,dt + \sigma_t\,dW_t,\qquad S_0\in L^2(\Omega,\mathcal F_0,\mathbb P),
\end{equation}
with $a,\sigma$ $\mathbb F$–progressively measurable and $\mathbb E\!\int_0^T |a_t|\,dt < \infty,
\mathbb E\!\int_0^T \sigma_t^2\,dt < \infty,
\sigma_t\ge 0\ \text{a.s.}$

We assume that $S$ is only \emph{partially observable} and the trader is only able to observe the process $Y$ given by 
\[
dY_t \;=\; h_t(S_t)\,dt \;+\; R_t^{1/2}\,dB_t,\qquad R_t\in[r_-,r_+] \text{ a.s. for some }0<r_-\le r_+<\infty,
\]
where $h$ is $\mathbb G$–progressively measurable with $h_t:\mathbb R\to\mathbb R$ and linear growth
$|h_t(x)|\le C(1+|x|)$ a.s.
Define the observation filtration $\mathbb G$ as the usual augmentation of
$\sigma(Y_s:0\le s\le t)$. Set
\[
\widehat h_t:=\mathbb E\!\left[h_t(S_t)\,\middle|\,\mathcal G_t\right].
\qquad
\widehat W_t\ :=\ \int_0^t R_s^{-1/2}\big(dY_s-\widehat h_s\,ds\big).
\]
Then $\widehat W$ is the innovation process, see \cite{BainCrisan, Xiong}, which is a $\mathbb G$–Brownian motion, and the following standard representation holds
\begin{equation}\label{eq:obs-innovation}
dY_t\ =\ \widehat h_t\,dt\ +\ R_t^{1/2}\,d\widehat W_t.
\end{equation}

Consider the (conditional–mean) filter $\widehat S_t:=\mathbb E[S_t\mid\mathcal G_t]$.
Under the integrability condition above, $\widehat S$ is an $\mathbb G$–continuous semimartingale with the \emph{innovation representation}
\begin{equation}\label{eq:filtered-dynamics}
d\widehat S_t \;=\; \widehat a_t\,dt \;+\; \widehat\sigma_t\,d\widehat W_t,
\end{equation}
where
\begin{equation}\label{eq:filtered-coeffs}
\widehat a_t \;:=\; \mathbb E[a_t\mid\mathcal G_t],
\qquad
\widehat\sigma_t \;:=\; \big(\mathbb E[S_t h_t\mid\mathcal G_t]-\mathbb E[S_t\mid\mathcal G_t]\mathbb E[h_t\mid\mathcal G_t]\big)R_t^{-1/2}.
\end{equation}
Writing $\pi_t(\varphi):=\mathbb E[\varphi(S_t)\mid\mathcal G_t]$, \eqref{eq:filtered-dynamics}–\eqref{eq:filtered-coeffs} follow from the standard Kushner–Stratonovich equation, see \cite{BainCrisan, Xiong},
\[
d\pi_t(\varphi)=\pi_t(a_t)\,dt+\big(\pi_t(\varphi h_t)-\pi_t(\varphi)\pi_t(h_t)\big)\,R_t^{-1/2}\,d\widehat W_t
\]
for the test function $\varphi(x)=x$.

Besides partial observability, we incorporate \emph{expert biases} by defining the expert-indexed proposal drifts $\rho^\lambda \coloneqq (\rho_t^\lambda)_{0\le t\le T}$ such that the real-valued map $(t, \omega, \lambda)\mapsto \rho^{\lambda}_t(w)$ is  $\mathcal P(\mathbb G)\otimes\mathcal B(\Lambda)$-measurable and $m=(m_t)_{0\le t\le T}$ is a $\mathbb G$-progressively measurable flow of probability kernels on $\Lambda$ which is some compact metric space representing the set of experts that, for simplicity, we take to be $[0,1]$, and let
\[
\bar\rho_t(m_t)\;:=\;\int_{\Lambda}\rho_t^\lambda\,m_t(d\lambda)
\]
be the \emph{aggregated bias/correction term} at time $t$, and $m$ is the aggregator which dynamically assigns weight to each expert's opinion. With a progressively measurable weight process $\beta=(\beta_t)_{0\le t\le T}$ taking values in $[0,1]$, we define the \emph{bias–adjusted filtered price} $\widetilde S$ by
\begin{equation}\label{eq:prelim-SDE}
d\widetilde S_t \;=\; \Big(\widehat a_t \;+\; \beta_t\big(\bar\rho_t(m_t)-\widehat a_t\big)\Big)\,dt \;+\; \widehat\sigma_t\,d\widehat W_t.
\end{equation}
Thus $\beta_t\equiv 0$ recovers the baseline filter \eqref{eq:filtered-dynamics}, while $\beta_t\equiv 1$ replaces the filtered drift by the aggregated bias $\bar\rho_t(m_t)$. Below, $\pi=(\pi_t)_{0\le t\le T}$ is an $\mathbb G$–progressively measurable flow of probability kernels on $\Lambda$ representing the prior beliefs of the trader regarding the experts. Note that the bias-adjusted price process can be written in the following equivalent form, emphasizing the fact that the \emph{trader forms a convex combination of the observed drift with the proposed correction term}:
\begin{equation}\label{eq:prelim-SDE}
d\widetilde S_t \;=\; \Big((1-\beta_t)\widehat a_t \;+\; \beta_t\bar\rho_t(m_t)\Big)\,dt \;+\; \widehat\sigma_t\,d\widehat W_t.
\end{equation}

We measure the discrepancy between $\bar\rho$ and $\widehat a$ by the functional
\begin{equation}\label{eq:prelim-L2-loss}
\mathcal L(m)\ :=\ \mathbb E\!\int_0^T \frac{\gamma}{2}\,\big(\bar\rho_t(m_t)-\widehat a_t\big)^2\,dt,
\qquad \gamma>0,
\end{equation}
which is convex in $m$ because $m\mapsto\bar\rho_t(m)$ is linear.

Finally, for flows of measures $m,\pi$ as above, we define the time–integrated relative entropy
\begin{equation}\label{eq:prelim-KL-distance}
\mathcal D_{\mathrm{KL}}(m\|\pi)
:= \mathbb E\!\int_0^T \mathrm{KL}\!\big(m_t\|\pi_t\big)\,dt
= \mathbb E\!\int_0^T\!\int_{\Lambda} \log\!\Big(\frac{dm_t}{d\pi_t}\Big)\,m_t(d\lambda)\,dt,
\end{equation}
with $\mathrm{KL}(m_t\|\pi_t)=+\infty$ if $m_t\not\ll\pi_t$. 

\begin{problem}\label{prob:KL}
Fix $0<K< \infty$. Minimize $\mathcal L(m)$ over $\mathbb G$–progressively measurable flows of probability kernels $m=(m_t)_{t\in[0,T]}$ subject to
\begin{equation}
\label{eq:KL-budget-constraint}
\mathcal D_{\mathrm{KL}}(m\|\pi)\ \le\ K,\qquad m_t\ll \pi_t\quad\text{a.s.\ for a.e.\ }t.
\end{equation}
\end{problem}

Thus, we study an optimization problem where the \emph{trader wants to select an optimal aggregator process} $m^{\ast}$ such that the posited drift is close to the filtered drift $\widehat a$, while keeping $m$ close to the given prior $\pi$. We make the following assumptions. 

\begin{assumption}\label{ass:data}
\leavevmode
\begin{enumerate}
\item[(A1)] For
each $(t,\omega)$, $\lambda\mapsto \rho_t^\lambda(\omega)$ is continuous and $(t,\omega,\lambda)\mapsto \rho_t^\lambda(\omega)$ is bounded. 
\item[(A2)] $\widehat a=(\widehat a_t)_{0\le t\le T}$ satisfies $\displaystyle \mathbb E \int_0^T |\widehat a_t|^2\,dt<\infty$.
\item[(A3)] For each fixed $t$, the kernel $\pi_t$ has full support. Since $\Lambda$ is compact and $\lambda\mapsto\rho_t^\lambda$ is bounded, the log–moment generating function
\begin{equation}
\label{eq:finite-mgf}
\Phi_t(\eta)\ :=\ \log\!\int_{\Lambda} e^{\eta\,\rho_t^\lambda}\,\pi_t(d\lambda)
\end{equation}
is finite for all $\eta\in\mathbb R$, a.s., for a.e.\ $t$.
\item[(A4)] $m=(m_t)_{0\le t\le T}$ satisfies $m_t\ll \pi_t$ a.s.\ for a.e.\ $t$.
\end{enumerate}
\end{assumption}

\begin{assumption}\label{ass:attainability}
For a.e.\ $(t,\omega)$, the filtered drift $\widehat a_t(\omega)$ belongs to the closed convex hull of the expert range:
\[
\widehat a_t(\omega)\ \in\ \overline{\mathrm{co}}\big\{\rho_t^\lambda(\omega):\lambda\in\Lambda\big\}
\ =\ \big[\mathrm{inf}_{\lambda}\rho_t^\lambda(\omega),\ \mathrm{sup}_{\lambda}\rho_t^\lambda(\omega)\big].
\]
\end{assumption}

\begin{problem}\label{prob:unconstrained}
Minimize $\mathcal L(m)$ in \eqref{eq:prelim-L2-loss} over \emph{all} $\mathbb F$–progressively measurable flows of probability kernels $m=(m_t)_{0\le t\le T}$ on $\Lambda$ (no absolute continuity requirement with respect to $\pi$).
\end{problem}

\begin{proposition}\label{prop:zero-loss-benchmark}
Under Assumption~\ref{ass:attainability}, Problem~\ref{prob:unconstrained} has minimal value $0$. Moreover, there exists an $\mathbb F$–progressively measurable flow $m^0=(m_t^0)_{t\in[0,T]}$ such that
\[
\bar\rho_t(m_t^0)\ =\ \int_{\Lambda}\rho_t^\lambda\,m_t^0(d\lambda)\ =\ \widehat a_t
\quad\text{a.s.\ for a.e.\ }t,
\]
so $\mathcal L(m^0)=0$.
\end{proposition}

\begin{remark}\label{rem:compatibility-largeK}
In Problem~\ref{prob:KL}, for each finite $K$ we keep the absolute continuity constraint $m_t\ll\pi_t$. As $K\uparrow\infty$, the KL constraint becomes asymptotically nonbinding and the optimizers may concentrate (in the limit) on sets of $\pi_t$–measure zero; thus Problem~\ref{prob:KL} $\Gamma$–converges to Problem~\ref{prob:unconstrained}, and the minimum loss approaches $0$.
\end{remark}

\begin{remark}\label{progressiveversions}
Basic facts regarding measurability properties of flows of probability kernels are collected in the appendix, where references are also given. In the following, we will use some techniques from the theory of relaxed controls, see e.g. \cite{CarmonaDelarue2}, \cite{Ahmed}, \cite{Fattorini}. In particular, once a $\mathbb G$-adapted relaxed control $Q(\omega, dt, d\lambda)$ is disintegrated into a product $dtQ(t,\omega, d\lambda)$, the Proposition 6.41, p.584, \cite{CarmonaDelarue2}, allows to choose a $\mathbb G$-progressively measurable version of $(t,\omega)\to Q(t,\omega, d\lambda)$, hence we do not always differentiate between adaptedness and progressiveness below.  
\end{remark}

\begin{lemma}\label{lem:KL-meas}
Equip $\mathcal P(\Lambda)$ with the weak topology. The map
$(\mu,\nu)\mapsto \mathrm{KL}(\mu\|\nu)$, extended by $+\infty$ when $\mu\not\ll\nu$,
is Borel and lower semicontinuous on $\mathcal P(\Lambda)\times\mathcal P(\Lambda)$.
Consequently, if $m,\pi$ are  $\mathbb G$-progressively measurable kernels, then
$(t,\omega)\mapsto \mathrm{KL}(m_t(\omega)\|\pi_t(\omega))$ is $\mathcal P(\mathbb G)$–measurable.
\end{lemma}

\begin{proof}
    See Tim van Erven, Peter Harremoës, \cite{ErvenHarremoes}, Theorem 19.
\end{proof}

For each $(t,\omega)$ the map $\mu\mapsto \mathrm{KL}(\mu\|\pi_t(\omega))$ is proper, strictly convex, and lower semicontinuous on $\mathcal P(\Lambda)$ (for the weak topology). Since (A1) ensures that $m\mapsto \bar\rho_t(m)$ is continuous and bounded on $\mathcal P(\Lambda)$, integrands below will be normal convex integrands in the sense of Rockafellar; this allows the interchange of infimum and integral used later, see \cite{Bonnans}, Section 3.2.

For $F:\mathcal P(\Lambda)\to\mathbb R$, the (linear) functional derivative $\frac{\delta F}{\delta m}(m,\lambda)$ is characterized by
\[
\frac{d}{d\varepsilon}F\big((1-\varepsilon)m+\varepsilon m'\big)\Big|_{\varepsilon=0}
=\int_{\Lambda}\frac{\delta F}{\delta m}(m,\lambda)\,(m'-m)(d\lambda),
\quad \forall\, m'\in\mathcal P(\Lambda),
\]
and we normalize it by $\int_\Lambda \frac{\delta F}{\delta m}(m,\lambda)\,m(d\lambda)=0$.
For the linear map $\bar\rho_t(m)=\int\rho_t^\lambda\,m(d\lambda)$ we have
\begin{equation}\label{eq:lions-bar-rho}
\frac{\delta \bar\rho_t}{\delta m}(m,\lambda)=\rho_t^\lambda-\bar\rho_t(m).
\end{equation}

\subsection{Existence of the Minimizing Flow}
\label{subsec:existence-duality}

\begin{theorem}\label{thm:existence and scalar}
Under Assumption~\ref{ass:data}, Problem~\ref{prob:KL} admits a solution $m^{\ast}$.
Moreover, there exists a scalar $\alpha^{\ast}\ge 0$ such that $m^{\ast}$ minimizes the Lagrangian
\begin{equation}
\label{eq:lagrangian}
\mathbb E\!\int_0^T \Big[\tfrac{\gamma}{2}\big(\bar\rho_t(m_t)-\widehat a_t\big)^2+\alpha^{\ast}\,\mathrm{KL}(m_t\|\pi_t)\Big]\,dt,
\end{equation}
and complementary slackness holds:
\begin{equation}
\label{eq:comp-slackness}
\alpha^{\ast}\Big(\mathcal D_{\mathrm{KL}}(m^{\ast}\|\pi)-K\Big)=0.
\end{equation}
\end{theorem}

\begin{proof}
We begin by setting up some notation. The proof is technical but is in fact a modification of standard arguments used for relaxed controls, which can be found in \cite{Ahmed, Fattorini}. Write the optimization over flows \(m=(m_t)_{t\in[0,T]}\) with the absolute continuity constraint \(m_t\ll \pi_t\) by means of the Radon–Nikodým densities
\[
f_t(\omega,\lambda)\ :=\ \frac{dm_t(\omega,\cdot)}{d\pi_t(\omega,\cdot)}(\lambda),\qquad f_t\ge 0,\qquad \int_{\Lambda} f_t(\omega,\lambda)\,\pi_t(\omega,d\lambda)=1\ \ \text{a.s.\ for a.e.\ }t .
\]
Set
\[
\mathsf X:=\Omega\times[0,T],\qquad \mathbb Q:=\mathbb P\otimes dt,\qquad
\Pi(d\omega,dt,d\lambda):=\mathbb P(d\omega)\,dt\,\pi_t(\omega,d\lambda).
\]

Note $\Pi$ is well-defined on $(\Omega\times[0,T]\times\Lambda, \mathcal P(\mathbb G)\otimes\mathcal B(\Lambda)$.

Then \(f\in L^1(\Pi)\) and the integrated KL constraint can be written as
\[
\mathcal D_{\mathrm{KL}}(m\|\pi)\,=\,\int f\log f\,d\Pi .
\]
Let \(M_t(\omega):=\sup_{\lambda\in\Lambda}|\rho_t^\lambda(\omega)|\); by (A1), \(M\in L^1(\mathbb Q)\) and \(|\bar\rho_t(m_t)|=\Big|\int \rho_t^\lambda f_t\,\pi_t(d\lambda)\Big|\le M_t\) a.s.\ for a.e.\ \(t\).
Finally, denote
\[
\mathcal A_K\ :=\ \Big\{f\in L^1(\Pi):\ f\ge 0,\ \int_{\Lambda} f_t(\omega,\lambda)\,\pi_t(\omega,d\lambda)=1\ \ \mathbb Q\text{-a.e.},\ \int f\log f\,d\Pi\le K\Big\}.
\]

Feasibility holds since \(f\equiv 1\in\mathcal A_K\) (this corresponds to \(m\equiv\pi\)), and \(\mathcal D_{\mathrm{KL}}(\pi\|\pi)=0<K\), i.e. Slater's condition is satisfied.

\medskip
We now show compactness of the feasible set. Let \(\{f^n\}\subset\mathcal A_K\). Since \(\int_\Lambda f_t^n\,d\pi_t=1\) \(\mathbb Q\)-a.e., we have \(\|f^n\|_{L^1(\Pi)}=\int f^n\,d\Pi=\int 1\,d\mathbb Q=\mathbb Q(\mathsf X)<\infty\) for all \(n\). Moreover, the KL bound \(\int f^n\log f^n\,d\Pi\le K\) gives uniform integrability of \(\{f^n\}\) by the de la Vallée Poussin criterion (use \(\Phi(u)=u\log u-u+1\), which is superlinear and \(\int \Phi(f^n)\,d\Pi\le K+\mathbb Q(\mathsf X)\)). Hence, by the Dunford--Pettis theorem, there is a subsequence (not relabelled) of \(\{f^n\}\) and \(f\in L^1(\Pi)\) with \(f^n\rightharpoonup f\) in \(L^1(\Pi)\), i.e. \(\mathcal A_K\) is relatively weakly compact in \(L^1(\Pi)\).

We next verify weak closedness. Let \(f^n\in\mathcal A_K\) with \(f^n\rightharpoonup f\) in \(L^1(\Pi)\).
Nonnegativity is preserved under \(L^1\)–weak limits. The normalization constraint passes to the limit because for every \(\varphi\in L^\infty(\mathsf X)\),
\[
\int_{\mathsf X}\!\varphi(\omega,t)\Big(\int_{\Lambda} f^n_t(\omega,\lambda)\,\pi_t(\omega,d\lambda)\Big)\,d\mathbb Q
=\int \varphi(\omega,t)\,f^n(\omega,t,\lambda)\,d\Pi
\ \longrightarrow\ 
\int \varphi\,f\,d\Pi,
\]
which implies \(\int f\,\pi_t(d\lambda)=1\) \(\mathbb Q\)-a.e. Likewise, the KL constraint is weakly lower semicontinuous: since \(u\mapsto u\log u\) is convex l.s.c.\ on \([0,\infty)\), the functional \(f\mapsto\int f\log f\,d\Pi\) is l.s.c.\ for the \(L^1\)–weak topology, hence \(\int f\log f\,d\Pi\le \liminf_n\int f^n\log f^n\,d\Pi\le K\). We also note that, by standard arguments, integrals of the limiting function $f$ with respect to $\pi_t$ are progressively measurable. Thus \(\mathcal A_K\) is weakly closed, hence weakly compact in \(L^1(\Pi)\).

\medskip
First we prove existence for bounded targets.
For \(n\in\mathbb N\), set the truncation \(\widehat a^{(n)}_t:=(-n)\vee(\widehat a_t\wedge n)\). Define the truncated functional
\[
\mathcal L_n(f)\ :=\ \mathbb E\!\int_0^T \frac{\gamma}{2}\Big(\textstyle\int_\Lambda \rho_t^\lambda f_t(\cdot,\lambda)\,\pi_t(d\lambda)-\widehat a^{(n)}_t\Big)^2 dt
=\frac{\gamma}{2}\int_{\mathsf X}\Big(S_f(\omega,t)-\widehat a^{(n)}_t(\omega)\Big)^2\,d\mathbb Q,
\]
where \(S_f(\omega,t):=\int_\Lambda \rho_t^\lambda(\omega) f_t(\omega,\lambda)\,\pi_t(\omega,d\lambda)\) (note this is just $\bar{\rho}_t$).
For each \(n\), we claim that \(\mathcal L_n\) is weakly lower semicontinuous on \(\mathcal A_K\). Define $\Phi_n(S_f)\coloneqq\mathcal L_n(f)$ and note that since $x\mapsto (x-\widehat a^{(n)})^2$ is convex and continuous in $x$, then by it is standard that $\Phi_n(S_f)$ is weakly lower semicontinuous as a functional of $S_f$, hence it suffices to show that \(S_{f^k}\rightharpoonup S_f\) in \(L^1(\mathsf X)\) as \(f^k\rightharpoonup f\) in \(L^1(\Pi)\). Indeed, if \(f^k\rightharpoonup f\) in \(L^1(\Pi)\), then for any \(\psi\in L^\infty(\mathsf X)\),
\[
\int_{\mathsf X}\psi S_{f^k}\,d\mathbb Q=\int \psi(\omega,t)\rho_t^\lambda(\omega)\,f^k(\omega,t,\lambda)\,d\Pi
\ \longrightarrow\ \int \psi(\omega,t)\rho_t^\lambda(\omega)\,f(\omega,t,\lambda)\,d\Pi
=\int_{\mathsf X}\psi S_f\,d\mathbb Q,
\]
because $(\omega,t,\lambda)\mapsto \psi(\omega,t)\rho_t^\lambda(\omega)\in L^\infty(\Pi)$ by (A1). Hence \(S_{f^k}\rightharpoonup S_f\) in \(L^1(\mathsf X)\).

By the above, \(\mathcal A_K\) is weakly compact. Since \(\mathcal L_n\) is weakly lower semicontinuous on \(\mathcal A_K\), there exists a minimizer \(f^{(n)}\in\mathcal A_K\) of \(\mathcal L_n\).

\medskip
We now remove the boundedness assumption. 
By Step~1, the sequence \(\{f^{(n)}\}_{n\ge 1}\subset\mathcal A_K\) is relatively weakly compact in \(L^1(\Pi)\). Extract a subsequence (not relabelled) such that \(f^{(n)}\rightharpoonup f^\ast\) in \(L^1(\Pi)\). Then \(f^\ast\in\mathcal A_K\) by the above.

For each fixed $f$, we have by elementary algebra
\[
\big|(S_f-\widehat a^{(n)})^2-(S_f-\widehat a)^2\big|
\le 2\,|S_f-\widehat a|\,|\widehat a-\widehat a^{(n)}|+(\widehat a-\widehat a^{(n)})^2.
\]
By (A1)–(A2), $S_f\in L^2(\mathsf X)$ (since $|S_f|\le M\in L^2$) and $\widehat a^{(n)}\to\widehat a$ in $L^2(\mathsf X)$. Hence the right-hand side converges to $0$ in $L^1(\mathsf X)$, and thus
\[
\mathcal L_n(f)\ \xrightarrow[n\to\infty]{}\ \mathcal L(f)
:=\frac{\gamma}{2}\int_{\mathsf X}\big(S_f-\widehat a\big)^2\,d\mathbb Q.
\]
Since $f^{(n)}\rightharpoonup f^\ast$ in $L^1(\Pi)$ and $f^\ast\in\mathcal A_K$, Fatou and the pointwise convergence of $\mathcal L_n$ yield $\mathcal L(f^\ast)=\inf_{\mathcal A_K}\mathcal L$.

Therefore, with \(\ell_n:=\inf_{g\in\mathcal A_K}\mathcal L_n(g)=\mathcal L_n(f^{(n)})\) and \(\ell:=\inf_{g\in\mathcal A_K}\mathcal L(g)\), we have \(\ell_n\downarrow L_\infty:=\lim_{n\to\infty}\ell_n\) and, by the previous result and Fatou,
\[
\mathcal L(f^\ast)\ \le\ \liminf_{n\to\infty}\mathcal L_n(f^{(n)})\ =\ L_\infty.
\]
On the other hand, for any \(g\in\mathcal A_K\), \(\mathcal L(g)=\lim_{n\to\infty}\mathcal L_n(g)\ge \lim_{n\to\infty}\ell_n=L_\infty\); taking the infimum over \(g\) yields \(\ell\ge L_\infty\). Combining the two expressions,
\[
\ell\ \le\ \mathcal L(f^\ast)\ \le\ L_\infty\ \le\ \ell,
\]
so \(\mathcal L(f^\ast)=\ell\). In terms of kernels \(m^\ast_t(\omega,d\lambda):=f^\ast_t(\omega,\lambda)\,\pi_t(\omega,d\lambda)\), this proves existence of an optimal adapted flow \(m^\ast\) for Problem~\ref{prob:KL}.

\medskip
Finally, we show strong duality and the existence of the Lagrange multiplier. The problem is convex (the map \(m\mapsto\bar\rho_t(m)\) is linear and \(x\mapsto x^2\) is convex; KL is convex) and proper, and Slater’s condition holds. Therefore, classical convex duality yields strong duality and existence of a Lagrange multiplier $\alpha^\ast\ge 0$ such that $m^\ast$ minimizes the Lagrangian
\[
\mathbb E\!\int_0^T\!\Big[\tfrac{\gamma}{2}\big(\bar\rho_t(m_t)-\widehat a_t\big)^2+\alpha^\ast\,\mathrm{KL}(m_t\|\pi_t)\Big]\,dt
\]
over all $\mathbb G$–progressively measurable flows $m$ with $m_t\ll\pi_t$ a.s.\ for a.e.\ $t$, and hence over all progressively measurable flows, since $\mathrm{KL}(m_t\|\pi_t)=+\infty$ when $m_t\not\ll\pi_t$. Moreover, complementary slackness holds:
\(\alpha^\ast\big(\mathcal D_{\mathrm{KL}}(m^\ast\|\pi)-K\big)=0\). 
\end{proof}

The Lagrangian integrand is a normal convex integrand in \((t,\omega,m_t)\) in the sense of Rockafellar, hence the minimization problem separates pointwise in \((t,\omega)\) into independent convex problems over \(\mathcal P(\Lambda)\). 
Fix $(t,\omega)$ and abbreviate $a=\widehat a_t(\omega)$, $\rho(\cdot)=\rho_t^\cdot(\omega)$, and $\pi=\pi_t(\omega)$.
The pointwise problem is
\begin{equation}
\label{eq:pointwise-problem}
\begin{aligned}
\min_{\mu\in\mathcal P(\Lambda),\,\mu\ll\pi}\ f(\mu)
&:=\frac{\gamma}{2}\Big(\int_{\Lambda}\rho\,d\mu - a\Big)^2+\alpha\,\mathrm{KL}(\mu\|\pi),
\qquad \alpha:=\alpha^{\ast}\ge 0.
\end{aligned}
\end{equation}

\subsection{Identification of the Form of the Minimizing Measures}
\label{subsec:first-variation}

Let $b(\mu):=\int \rho\,d\mu$.
By equation~\ref{eq:lions-bar-rho} and the Gâteaux derivative of KL, the Lions first variation of $f$ at $\mu$ is
\begin{equation}
\label{eq:first-variation}
\frac{\delta f}{\delta \mu}(\mu,\lambda)
=\gamma\big(b(\mu)-a\big)\,\rho(\lambda)
\ +\ \alpha\Big(\log\tfrac{d\mu}{d\pi}(\lambda)+1\Big)\ +\ \mathrm{const},
\end{equation}
with the additive constant fixed by normalization.

At any minimizer $\mu^{\ast}$ (necessarily with strictly positive density when $\alpha>0$), the KKT optimality condition yields
\begin{equation}
\label{eq:euler-lagrange}
\gamma\big(b^{\ast}-a\big)\,\rho(\lambda)
+\alpha\Big(\log\tfrac{d\mu^{\ast}}{d\pi}(\lambda)+1\Big)
=\mathrm{const},
\qquad b^{\ast}:=\int \rho\,d\mu^{\ast}.
\end{equation}

\begin{proposition}\label{prop:gibbs-solution}
If $\alpha>0$, the unique minimizer of \eqref{eq:pointwise-problem} is the Gibbs measure
\begin{equation}
\label{eq:gibbs-tilt}
\begin{aligned}
d\mu^{\ast}_\theta(\lambda)
&=\frac{e^{-\theta\,\rho(\lambda)}}{Z(\theta)}\,\pi(d\lambda),\\
Z(\theta)
&:=\int_{\Lambda} e^{-\theta\rho(\lambda)}\,\pi(d\lambda),\\
\theta
&=\frac{\gamma}{\alpha}\,(b^{\ast}-a),
\end{aligned}
\end{equation}
with mean
\begin{equation}
\label{eq:psi}
\begin{aligned}
\psi(\theta)
&:=\int \rho\,d\mu^{\ast}_\theta
=\frac{\displaystyle\int_{\Lambda} \rho(\lambda) e^{-\theta\rho(\lambda)}\,\pi(d\lambda)}
{\displaystyle\int_{\Lambda} e^{-\theta\rho(\lambda)}\,\pi(d\lambda)}
=-\frac{d}{d\theta}\log Z(\theta).
\end{aligned}
\end{equation}
The scalar $\theta$ is the unique solution of
\begin{equation}
\label{eq:theta-fixed-pt}
\psi(\theta)=a+\frac{\alpha}{\gamma}\,\theta.
\end{equation}
\end{proposition}

\begin{proof}
Fix $(t,\omega)$ and set $\rho:=\rho_t^\cdot(\omega)$, $\pi:=\pi_t(\omega)$ and $a:=\widehat a_t(\omega)$; we suppress the indices for clarity. For $\mu\ll\pi$ set
\[
F(\mu):=\frac{\gamma}{2}\big(b(\mu)-a\big)^2+\alpha\,\mathrm{KL}(\mu\|\pi),
\qquad b(\mu):=\int_\Lambda \rho\,d\mu.
\]
Assume $\alpha>0$. By the boundedness of $\rho$ (Assumption~\ref{ass:data}) the partition function
\[
Z(\theta):=\int_{\Lambda} e^{-\theta\rho(\lambda)}\,\pi(d\lambda)\in(0,\infty),\qquad \theta\in\mathbb R,
\]
is finite and $C^1$ in $\theta$, with
\[
\psi(\theta):=\int \rho\,d\mu_\theta^\ast
=\frac{\int \rho e^{-\theta\rho}\,d\pi}{\int e^{-\theta\rho}\,d\pi}
=-\frac{d}{d\theta}\log Z(\theta),
\]
where $\mu_\theta^\ast$ is defined by \eqref{eq:gibbs-tilt}.

\smallskip

The variational formula for KL divergence, see, e.g., \cite{DupuisEllis}, p.29, Lemma 1.4.3, asserts that for any bounded measurable $\varphi$ and any $\mu\ll\pi$,
\[
\mathrm{KL}(\mu\|\pi)\ \ge\ \int \varphi\,d\mu-\log\!\int e^{\varphi}\,d\pi,
\]
with equality iff $\frac{d\mu}{d\pi}\propto e^{\varphi}$. Applying this with $\varphi=-\theta\rho$ gives, for all $\theta\in\mathbb R$ and $\mu\ll\pi$,
\begin{equation}\label{eq:entropy-ineq}
\mathrm{KL}(\mu\|\pi)\ \ge\ -\theta\,b(\mu)-\log Z(\theta),
\end{equation}
with equality iff $\mu=\mu_\theta^\ast$ as in \eqref{eq:gibbs-tilt}. Hence
\[
F(\mu)\ \ge\ \alpha\big(-\theta\,b(\mu)-\log Z(\theta)\big)+\frac{\gamma}{2}\big(b(\mu)-a\big)^2
=: \Phi_\theta\big(b(\mu)\big).
\]

For fixed $\theta$, $\Phi_\theta(x)$ is a strictly convex quadratic in $x$ with unique minimizer
\[
x^\ast(\theta)=a+\frac{\alpha}{\gamma}\,\theta,
\]
and minimal value
\begin{equation}\label{eq:theta-lb}
\inf_{x\in\mathbb R}\Phi_\theta(x)
=-\alpha\Big(\theta a+\log Z(\theta)\Big)-\frac{\alpha^2}{2\gamma}\,\theta^2.
\end{equation}
Therefore, for all $\mu\ll\pi$ and $\theta\in\mathbb R$,
\begin{equation}\label{eq:global-lb}
F(\mu)\ \ge\ -\alpha\Big(\theta a+\log Z(\theta)\Big)-\frac{\alpha^2}{2\gamma}\,\theta^2.
\end{equation}

Equality in \eqref{eq:global-lb} for some $\theta$ occurs iff simultaneously:
(i) equality holds in \eqref{eq:entropy-ineq}, i.e.\ $\mu=\mu_\theta^\ast$, and
(ii) $b(\mu)=x^\ast(\theta)$, i.e.\ $b(\mu_\theta^\ast)=a+\frac{\alpha}{\gamma}\theta$.
Since $b(\mu_\theta^\ast)=\psi(\theta)$, condition (ii) is equivalent to the fixed–point equation
\begin{equation}\label{eq:fixed-point-proof}
\psi(\theta)=a+\frac{\alpha}{\gamma}\,\theta,
\end{equation}
which is precisely \eqref{eq:theta-fixed-pt}. If $\theta^\ast$ solves \eqref{eq:fixed-point-proof}, then $\mu_{\theta^\ast}^\ast$ achieves equality in \eqref{eq:global-lb} and is therefore optimal for the primal problem, with optimizer of the form \eqref{eq:gibbs-tilt} and tilted mean given by \eqref{eq:psi}.

Because $\rho$ is bounded, $Z$ is $C^2$ and
\[
\psi'(\theta)=-\frac{d^2}{d\theta^2}\log Z(\theta)=-\mathrm{Var}_{\mu_\theta^\ast}(\rho)\le 0.
\]
Hence the function $g(\theta):=\psi(\theta)-a-\frac{\alpha}{\gamma}\theta$ is strictly decreasing:
\[
g'(\theta)=\psi'(\theta)-\frac{\alpha}{\gamma}\le -\frac{\alpha}{\gamma}<0.
\]
Moreover, $\psi(\theta)$ is bounded while the linear term diverges, so $g(\theta)\to +\infty$ as $\theta\to -\infty$ and $g(\theta)\to -\infty$ as $\theta\to +\infty$. By continuity there exists a unique $\theta^\ast$ solving \eqref{eq:fixed-point-proof}. For $\alpha>0$, $F$ is strictly convex on $\{\mu\ll\pi\}$ (sum of the strictly convex $\alpha\,\mathrm{KL}(\cdot\|\pi)$ and a convex function of the linear statistic $b(\mu)$), hence the minimizer is unique. Since $\mu_{\theta^\ast}^\ast$ is feasible and optimal, it is the unique minimizer.

This proves that for $\alpha>0$ the unique optimizer is the Gibbs measure \eqref{eq:gibbs-tilt}; its mean is $\psi(\theta)$ in \eqref{eq:psi}, and the parameter $\theta$ is uniquely determined by \eqref{eq:theta-fixed-pt}.
\end{proof}

\begin{lemma}\label{lem:unique-theta}
Under Assumption~\ref{ass:data}\,(A3), $\psi$ in \eqref{eq:psi} is continuous and nonincreasing on $\mathbb R$, with range
$\psi(\mathbb R)=[\mathrm{inf}_\lambda \rho(\lambda),\ \mathrm{sup}_\lambda \rho(\lambda)]$.
Hence \eqref{eq:theta-fixed-pt} admits a unique solution $\theta\in\mathbb R$.
\end{lemma}

\begin{proof}
Fix $(t,\omega)$ and set $\rho:=\rho_t^\cdot(\omega)$ and $\pi:=\pi_t(\omega)$, as well as
\[
Z(\theta):=\int_\Lambda e^{-\theta\,\rho(\lambda)}\,\pi(d\lambda),\qquad
d\mu_\theta(\lambda):=\frac{e^{-\theta\,\rho(\lambda)}}{Z(\theta)}\,\pi(d\lambda),\qquad
\psi(\theta):=\int_\Lambda \rho\,d\mu_\theta.
\]
Write $m:=\mathrm{ess\,inf}_\lambda \rho(\lambda)$ and $M:=\mathrm{ess\,sup}_\lambda \rho(\lambda)$ (finite due to the assumptions).

By Assumption~\ref{ass:data}\,(A3), $\Phi(\eta):=\log\!\int e^{\eta \rho}\,d\pi$ is finite for all $\eta\in\mathbb R$. Hence $Z(\theta)=e^{\Phi(-\theta)}\in(0,\infty)$ for every $\theta\in\mathbb R$. Moreover, on any compact interval $I\subset\mathbb R$ there exists $c=\sup_{\theta\in I}|\theta|<\infty$ and, for any $\delta>0$,
\[
|\,\rho\,|e^{-\theta\rho}\ \le\ \tfrac1\delta e^{(c+\delta)|\rho|}\ \le\ \tfrac1\delta\big(e^{(c+\delta)\rho}+e^{-(c+\delta)\rho}\big),\qquad
\rho^2 e^{-\theta\rho}\ \le\ \tfrac{2}{\delta^2} e^{(c+\delta)|\rho|},
\]
so Assumption~\ref{ass:data}\,(A3) implies the right-hand sides are $\pi$–integrable. Dominated convergence then yields $Z\in C^2(\mathbb R)$ with
\[
Z'(\theta)=-\!\int \rho\,e^{-\theta\rho}\,d\pi,\qquad
Z''(\theta)=\!\int \rho^2 e^{-\theta\rho}\,d\pi.
\]
Let $A(\theta):=\log Z(\theta)$. Then $A'(\theta)=Z'(\theta)/Z(\theta)=-\psi(\theta)$ and
\[
A''(\theta)\ =\ \frac{Z''(\theta)}{Z(\theta)}-\Big(\frac{Z'(\theta)}{Z(\theta)}\Big)^2\ =\ \mathrm{Var}_{\mu_\theta}(\rho)\ \ge\ 0,
\]
so $\psi(\theta)=-A'(\theta)$ is continuous and nonincreasing on $\mathbb R$. If $\rho$ is not $\pi$–a.s.\ constant, then $\mu_\theta$ is equivalent to $\pi$ (its density $e^{-\theta\rho}/Z(\theta)$ is strictly positive $\pi$–a.s.), hence $\rho$ is not $\mu_\theta$–a.s.\ constant and $\mathrm{Var}_{\mu_\theta}(\rho)>0$ for all $\theta$, i.e.\ $\psi$ is strictly decreasing.

For every $\theta$, since $\mu_\theta$ is a probability measure and $\rho\in L^1(\mu_\theta)$,
\[
m\ \le\ \psi(\theta)\ \le\ M.
\]
We show $\lim_{\theta\to+\infty}\psi(\theta)=m$; the case $\theta\to-\infty$ is analogous (with limit $M$). Fix $\varepsilon>0$ and set
\[
A_\varepsilon:=\{\lambda:\rho(\lambda)\le m+\varepsilon\},\qquad B_\varepsilon:=\Lambda\setminus A_\varepsilon=\{\lambda:\rho(\lambda)> m+\varepsilon\}.
\]
By definition of $m=\mathrm{inf}_\lambda \rho(\lambda)$, $\pi(A_\varepsilon)>0$ due to (A4). Decompose, for $\theta>0$,
\[
Z(\theta)=\!\int e^{-\theta\rho}\,d\pi
= e^{-\theta(m+\varepsilon)}\!\left(\underbrace{\int_{A_\varepsilon} e^{-\theta(\rho-m-\varepsilon)}\,d\pi}_{=:C_\varepsilon(\theta)}
+\underbrace{\int_{B_\varepsilon} e^{-\theta(\rho-m-\varepsilon)}\,d\pi}_{=:L_\varepsilon(\theta)}\right),
\]
and
\[
\int \rho\,e^{-\theta\rho}\,d\pi
= e^{-\theta(m+\varepsilon)}\!\left(\int_{A_\varepsilon} \rho\,e^{-\theta(\rho-m-\varepsilon)}\,d\pi
+\int_{B_\varepsilon} \rho\,e^{-\theta(\rho-m-\varepsilon)}\,d\pi\right).
\]
Since $\rho\le m+\varepsilon$ on $A_\varepsilon$, we have
\[
\int_{A_\varepsilon} \rho\,e^{-\theta(\rho-m-\varepsilon)}\,d\pi\ \le\ (m+\varepsilon)\,C_\varepsilon(\theta).
\]
On $B_\varepsilon$ we have $\rho-m-\varepsilon>0$, hence $e^{-\theta(\rho-m-\varepsilon)}\downarrow 0$ pointwise as $\theta\to\infty$ and is dominated by $1$. Because $\rho\in L^1(\pi)$ (by Assumption~\ref{ass:data}\, using $|x|\le \delta^{-1}e^{\delta |x|}$ and two-sided exponential integrability), dominated convergence gives
\[
L_\varepsilon(\theta)\ \longrightarrow\ 0
\quad\text{and}\quad
R_\varepsilon(\theta):=\int_{B_\varepsilon} \rho\,e^{-\theta(\rho-m-\varepsilon)}\,d\pi\ \longrightarrow\ 0
\qquad\text{as }\theta\to\infty.
\]
Therefore,
\[
\psi(\theta)=\frac{\int \rho\,e^{-\theta\rho}\,d\pi}{Z(\theta)}
\ \le\ \frac{(m+\varepsilon)\,C_\varepsilon(\theta)+R_\varepsilon(\theta)}{C_\varepsilon(\theta)+L_\varepsilon(\theta)}
\ \xrightarrow[\theta\to\infty]{}\ m+\varepsilon.
\]
Since $\varepsilon>0$ is arbitrary and $\psi(\theta)\ge m$, it follows that $\lim_{\theta\to\infty}\psi(\theta)=m$. The limit $\lim_{\theta\to-\infty}\psi(\theta)=M$ is proved symmetrically by applying the same argument to the sets $A_\varepsilon^{+}:=\{\rho\ge M-\varepsilon\}$ and $B_\varepsilon^{+}:=\{\rho<M-\varepsilon\}$. Consequently,
\[
\psi(\mathbb R)=[m,M]=\big[\mathrm{inf}_\lambda \rho(\lambda),\ \mathrm{sup}_\lambda \rho(\lambda)\big].
\]

The map $\theta\mapsto \psi(\theta)$ is continuous and bounded, with $\psi(\theta)\to M$ as $\theta\to-\infty$ and $\psi(\theta)\to m$ as $\theta\to+\infty$. If $\alpha>0$, then
\[
g(\theta):=\psi(\theta)-a-\frac{\alpha}{\gamma}\theta
\]
is strictly decreasing, with $\lim_{\theta\to-\infty}g(\theta)=+\infty$ and $\lim_{\theta\to+\infty}g(\theta)=-\infty$, hence there exists a unique $\theta\in\mathbb R$ solving \eqref{eq:theta-fixed-pt}. If $\alpha=0$, then \eqref{eq:theta-fixed-pt} reduces to $\psi(\theta)=a$, which has a solution iff $a\in[m,M]$; it is unique when $\rho$ is not $\pi$–a.s.\ constant (since $\psi$ is strictly decreasing), and either has no solution or infinitely many solutions when $\rho$ is $\pi$–a.s.\ constant (according as $a\neq m=M$ or $a=m=M$).

\end{proof}

\begin{lemma}\label{lem:monotonicity-alpha}
For $\alpha_2>\alpha_1>0$, the corresponding optimizers $\mu^{\ast}_{\alpha_i}$ satisfy
\begin{equation*}
\mathrm{KL}(\mu^{\ast}_{\alpha_2}\|\pi)\ \le\ \mathrm{KL}(\mu^{\ast}_{\alpha_1}\|\pi).
\end{equation*}
\end{lemma}


\begin{proof}
Let $F_\alpha(\mu):=\frac{\gamma}{2}\big(b(\mu)-a\big)^2+\alpha\,\mathrm{KL}(\mu\|\pi)$ and let $\mu^*_{\alpha_i}$ be any minimizer of $F_{\alpha_i}$ for $i=1,2$. By boundedness of $\rho$ we have $F_{\alpha_i}(\pi)<\infty$, hence $F_{\alpha_i}(\mu^*_{\alpha_i})\le F_{\alpha_i}(\pi)<\infty$, which in particular implies $\mathrm{KL}(\mu^*_{\alpha_i}\|\pi)<\infty$.

By optimality,
\begin{align}
F_{\alpha_2}(\mu^*_{\alpha_2}) &\le F_{\alpha_2}(\mu^*_{\alpha_1}), \label{eq:opt-alpha2}\\
F_{\alpha_1}(\mu^*_{\alpha_1}) &\le F_{\alpha_1}(\mu^*_{\alpha_2}). \label{eq:opt-alpha1}
\end{align}
Expanding \eqref{eq:opt-alpha2}--\eqref{eq:opt-alpha1} yields
\begin{align*}
\frac{\gamma}{2}\big(b(\mu^*_{\alpha_2})-a\big)^2+\alpha_2\,\mathrm{KL}(\mu^*_{\alpha_2}\|\pi)
&\le \frac{\gamma}{2}\big(b(\mu^*_{\alpha_1})-a\big)^2+\alpha_2\,\mathrm{KL}(\mu^*_{\alpha_1}\|\pi),\\
\frac{\gamma}{2}\big(b(\mu^*_{\alpha_1})-a\big)^2+\alpha_1\,\mathrm{KL}(\mu^*_{\alpha_1}\|\pi)
&\le \frac{\gamma}{2}\big(b(\mu^*_{\alpha_2})-a\big)^2+\alpha_1\,\mathrm{KL}(\mu^*_{\alpha_2}\|\pi).
\end{align*}
Adding the two inequalities and cancelling the matching quadratic terms gives
\[
\alpha_2\,\mathrm{KL}(\mu^*_{\alpha_2}\|\pi)+\alpha_1\,\mathrm{KL}(\mu^*_{\alpha_1}\|\pi)
\ \le\
\alpha_2\,\mathrm{KL}(\mu^*_{\alpha_1}\|\pi)+\alpha_1\,\mathrm{KL}(\mu^*_{\alpha_2}\|\pi),
\]
i.e.
\[
(\alpha_2-\alpha_1)\Big(\mathrm{KL}(\mu^*_{\alpha_2}\|\pi)-\mathrm{KL}(\mu^*_{\alpha_1}\|\pi)\Big)\ \le\ 0.
\]
Since $\alpha_2>\alpha_1>0$, it follows that $\mathrm{KL}(\mu^*_{\alpha_2}\|\pi)\le \mathrm{KL}(\mu^*_{\alpha_1}\|\pi)$, as claimed.
\end{proof}

\subsection{Optimal Adapted Measure-Valued Control}
\label{subsec:assemble}

Define $\theta_t^{\ast}(\omega)$ to be the unique solution of \eqref{eq:theta-fixed-pt} with $a=\widehat a_t(\omega)$, $\rho=\rho_t^\cdot(\omega)$, and $\pi=\pi_t(\omega)$.
Set
\begin{equation}
\label{eq:optimal-m-process}
m_t^{\ast}(d\lambda)\ :=\
\frac{\exp\!\big(-\theta_t^{\ast}(\omega)\,\rho_t^\lambda(\omega)\big)}{\displaystyle \int_{\Lambda}\exp\!\big(-\theta_t^{\ast}(\omega)\,\rho_t^{\lambda'}(\omega)\big)\,\pi_t(\omega,d\lambda')}
\ \pi_t(\omega,d\lambda).
\end{equation}

\begin{proposition}\label{prop:adaptedness} $(\theta_t^{\ast})_{t\in[0,T]}$ is a $\mathbb G$–adapted real–valued process. Consequently, for every $t\in[0,T], B\in\mathcal B(\Lambda)$, the map $\omega\mapsto m_t^{\ast}(\omega,B)$ is $\mathcal G_t$–measurable, so $(m_t^{\ast})_{t\in[0,T]}$ is an $\mathbb G$–adapted flow of probability kernels. Moreover, $m^{\ast}$ is optimal for Problem~\ref{prob:KL}.

If, in addition, the maps $(t,\omega)\mapsto\rho_t^\lambda(\omega)$, $(t,\omega)\mapsto\widehat a_t(\omega)$, and $(t,\omega)\mapsto\pi_t(\omega,B)$ are progressively measurable for each fixed $\lambda\in\Lambda$ and $B\in\mathcal B(\Lambda)$, then the maps $(t,\omega)\mapsto\theta_t^{\ast}(\omega)$ and $(t,\omega)\mapsto m_t^{\ast}(\omega,B)$ are $\mathcal P(\mathbb G)$–measurable (i.e.\ progressively measurable).
\end{proposition}


\begin{proof}
Fix $t\in[0,T]$ and set $\rho(\omega,\lambda):=\rho_t^\lambda(\omega)$, $\pi(\omega,\cdot):=\pi_t(\omega,\cdot)$ and $a(\omega):=\widehat a_t(\omega)$. For each $\theta\in\mathbb R$ set
\[
Z(\omega,\theta):=\int_\Lambda e^{-\theta\,\rho(\omega,\lambda)}\,\pi(\omega,d\lambda),\qquad
\psi(\omega,\theta):=\frac{\displaystyle\int_\Lambda \rho(\omega,\lambda)\,e^{-\theta\,\rho(\omega,\lambda)}\,\pi(\omega,d\lambda)}{Z(\omega,\theta)}.
\]

We begin by showing the measurability of $(\omega,\theta)\mapsto \psi(\omega,\theta)$ at fixed $t$.
From assumptions it follows that $(\omega,\lambda)\mapsto \rho(\omega,\lambda)$ is $\mathcal G_t\otimes\mathcal B(\Lambda)$–measurable. Consequently, for each fixed $\theta$, the functions
\[
(\omega,\lambda)\mapsto e^{-\theta\,\rho(\omega,\lambda)}
\quad\text{and}\quad
(\omega,\lambda)\mapsto \rho(\omega,\lambda)\,e^{-\theta\,\rho(\omega,\lambda)}
\]
are $\mathcal G_t\otimes\mathcal B(\Lambda)$–measurable and bounded. Since $\omega\mapsto\pi(\omega,\cdot)$ is an $\mathcal G_t$–measurable kernel on $\Lambda$, the map $\omega\mapsto \int_\Lambda g(\omega,\lambda)\,\pi(\omega,d\lambda)$ is $\mathcal G_t$–measurable for any $\mathcal G_t$–measurable $\omega\mapsto g(\omega,\lambda)$, see \cite{Kallenberg}, Lemma 1.41. Therefore, for each $\theta$, both $\omega\mapsto Z(\omega,\theta)$ and $\omega\mapsto \int_{\Lambda} \rho(\omega,\lambda) e^{-\theta\rho(\omega,\lambda)}\,d\pi(\omega,d\lambda)$ are $\mathcal G_t$–measurable; as $Z(\omega,\theta)>0$, the ratio $\psi(\omega,\theta)$ is $\mathcal G_t$–measurable. Moreover, by boundedness of $\rho$ and dominated convergence, $\theta\mapsto \psi(\omega,\theta)$ is continuous for each $\omega$, hence a Caratheodory function, hence by Lemma 4.51, \cite{AliprantisBorder}, p.153, it is $\mathcal G_t\otimes\mathcal B(\mathbb R)$–measurable.

Then, we show $\mathcal G_t$–measurability of the fixed point $\theta_t^{\ast}$.
Assume $\alpha>0$. Define
\[
f(\omega,\theta):=\psi(\omega,\theta)-a(\omega)-\frac{\alpha}{\gamma}\,\theta.
\]
By the previous step and measurability of $a(\cdot)$, for each fixed $\theta$ the map $\omega\mapsto f(\omega,\theta)$ is $\mathcal G_t$–measurable, and for each $\omega$, $\theta\mapsto f(\omega,\theta)$ is continuous and strictly decreasing (Lemma~\ref{lem:unique-theta}). For each $\omega$ there is a unique root $\theta_t^{\ast}(\omega)$ of $f(\omega,\theta)=0$. Using monotonicity,
\[
\theta_t^{\ast}(\omega)\ =\ \inf\big\{q\in\mathbb Q:\ f(\omega,q)\le 0\big\}.
\]
Hence, for any $r\in\mathbb R$,
\[
\{\omega:\ \theta_t^{\ast}(\omega)<r\}
=\bigcup_{q\in\mathbb Q,\ q<r}\{\omega:\ f(\omega,q)\le 0\}\in\mathcal G_t,
\]
since each $\omega\mapsto f(\omega,q)$ is $\mathcal G_t$–measurable. Thus $\omega\mapsto \theta_t^{\ast}(\omega)$ is $\mathcal G_t$–measurable.

If $\alpha^*=0$ and, in addition, 
$\widehat a_t(\omega)\in\big[\mathrm{inf}_\lambda \rho(\lambda),\ \mathrm{sup}_\lambda \rho(\lambda)\big]
\quad\text{for a.e.\ $(t,\omega)$},$ then one can define $\theta_t^*(\omega)$ as any (measurable) solution of $\psi_t(\omega,\theta)=\widehat a_t(\omega)$ (e.g., the minimal root), and the conclusions above still hold. 

To show the adaptedness of $m_t^{\ast}$, define
\[
h_t(\omega,\lambda):=\exp\!\big(-\theta_t^{\ast}(\omega)\,\rho(\omega,\lambda)\big).
\]
Note $h_t(\omega,\lambda)$ is $\mathcal G_t$-measurable in $\omega$ and continuous in $\lambda$, hence a Caratheodory function, hence the map $(\omega,\lambda)\mapsto h_t(\omega,\lambda)$ is $\mathcal G_t\otimes\mathcal B(\Lambda)$–measurable, and strictly positive. Therefore, for each $B\in\mathcal B(\Lambda)$ the numerator
\[
\omega\ \longmapsto\ \int_B h_t(\omega,\lambda)\,\pi(\omega,d\lambda)
\]
and the (strictly positive) denominator
\[
\omega\ \longmapsto\ \int_\Lambda h_t(\omega,\lambda)\,\pi(\omega,d\lambda)
\]
are $\mathcal G_t$–measurable. Their ratio equals $\omega\mapsto m_t^{\ast}(\omega,B)$ by \eqref{eq:optimal-m-process}, hence $m_t^{\ast}$ is $\mathcal G_t$–measurable for each $t$, i.e., $(m_t^{\ast})_{t\in[0,T]}$ is an $\mathbb G$–adapted flow of kernels.

Finally, we must show that the obtained controls are optimal. Let $\alpha^{\ast}$ be the scalar from Theorem~\ref{thm:existence and scalar}. For this choice, the above gives $\theta_t^{\ast}$ and that $m^{\ast}$ is adapted. By Proposition~\ref{prop:gibbs-solution}, for each $(t,\omega)$ the measure $m_t^{\ast}(\omega,\cdot)$ solves the corresponding pointwise minimization of the Lagrangian integrand, hence $m^{\ast}$ minimizes \eqref{eq:lagrangian}. Theorem~\ref{thm:existence and scalar} then implies $m^{\ast}$ is optimal for Problem~\ref{prob:KL}.
\end{proof}

\begin{lemma}\label{lem:KL-budget}
At $(t,\omega)$ and $\theta=\theta_t^{\ast}(\omega)$,
\begin{equation}
\label{eq:KL-value}
\mathrm{KL}\big(m_t^{\ast}\|\pi_t\big)
=-\theta\,\psi_t(\theta)-\log Z_t(\theta),
\qquad Z_t(\theta):=\int_{\Lambda} e^{-\theta\rho_t^\lambda}\,\pi_t(d\lambda).
\end{equation}
The map $\alpha\mapsto \mathbb E\!\int_0^T \mathrm{KL}(m_t^{\ast}(\alpha)\|\pi_t)\,dt$ is continuous and nonincreasing on $(0,\infty)$. It is strictly decreasing provided
\[
\mathbb Q\big\{(t,\omega): a_t(\omega)\neq \psi_t(0,\omega)\big\}>0
\quad\text{and}\quad
\rho_t^\cdot(\omega)\ \text{is not }\pi_t(\omega)\text{-a.s.\ constant a.e.}
\]
Hence there exists at least one $\alpha^{\ast}\ge 0$ enforcing the global constraint:
\begin{equation*}
\mathbb E\!\int_0^T \mathrm{KL}(m_t^{\ast}\|\pi_t)\,dt\ =\ K,
\end{equation*}
with $\alpha^{\ast}>0$ when $0<K< G(0+):=\lim_{\alpha\downarrow 0}\mathbb E\!\int_0^T \mathrm{KL}(m_t^{\ast}(\alpha)\|\pi_t)\,dt$, and $\alpha^{\ast}=0$ when $K\ge G(0+)$. If, in addition, the strictness condition holds, then this $\alpha^{*}$ is unique.

\end{lemma}

\begin{proof}
We omit the dependence on $t$ for simplicity. Since $\frac{dm^{\ast}_\theta}{d\pi}(\lambda)=e^{-\theta\,\rho(\lambda)}/Z(\theta)$, we have
\[
\log\!\Big(\frac{dm^{\ast}_\theta}{d\pi}\Big)=-\theta\,\rho-\log Z(\theta),
\]
and therefore
\[
\mathrm{KL}(m^{\ast}_\theta\|\pi)=\int \log\!\Big(\frac{dm^{\ast}_\theta}{d\pi}\Big)\,dm^{\ast}_\theta
=-\theta\!\int \rho\,dm^{\ast}_\theta-\log Z(\theta)
=-\theta\,\psi(\theta)-\log Z(\theta).
\]
This yields \eqref{eq:KL-value}.

Recall that by Lemma~\ref{lem:monotonicity-alpha} $\alpha\mapsto \mathrm{KL}(\mu^{\ast}_{\alpha}\|\pi)$ is nonincreasing pointwise in $(t,\omega)$. Consequently, $\alpha\mapsto \mathbb E\!\int_0^T \mathrm{KL}(m_t^{\ast}(\alpha)\|\pi_t)\,dt$ is nonincreasing. It is strictly decreasing if 
\[
\mathbb Q\big\{(t,\omega): a_t(\omega)\neq \psi_t(0,\omega)\big\}>0
\ \text{and}\ \rho_t^\cdot(\omega)\ \text{is not }\pi_t(\omega)\text{-a.s.\ constant a.e.}
\]

To show continuity of $\alpha\mapsto \mathbb E\!\int_0^T \mathrm{KL}(m_t^{\ast}(\alpha)\|\pi_t)\,dt$ on $(0,\infty)$, fix $\alpha_0>0$ and work pointwise in $(t,\omega)$.
Let $\theta(\alpha)$ denote the unique solution of the fixed–point equation
\[
\psi(\theta)=a+\frac{\alpha}{\gamma}\,\theta .
\]
Set $f(\theta,\alpha):=\psi(\theta)-a-(\alpha/\gamma)\theta$. Then $f(\theta(\alpha),\alpha)=0$ and
\[
\partial_\theta f(\theta,\alpha)=\psi'(\theta)-\frac{\alpha}{\gamma}=-\mathrm{Var}_{m^{\ast}_\theta}(\rho)-\frac{\alpha}{\gamma}\ <\ 0.
\]
By the (deterministic) implicit function theorem, $\alpha\mapsto \theta(\alpha)$ is $C^1$ in a neighborhood of $\alpha_0$; hence
\[
\alpha\ \longmapsto\ \mathrm{KL}\big(m^{\ast}_{\theta(\alpha)}\|\pi\big)
=-\theta(\alpha)\,\psi(\theta(\alpha))-\log Z(\theta(\alpha))
\]
is continuous pointwise in $(t,\omega)$.

To pass to expectation and time–integration, observe that for all $\alpha$ and any competitor $\mu$,
\[
\mathrm{KL}\big(\mu^{\ast}_\alpha\|\pi\big)\ \le\ \frac{1}{\alpha}\,F_\alpha(\mu).
\]
Choosing $\mu=\mu_{\alpha_0}^{\ast}$ and restricting $\alpha\in[\alpha_0/2,2\alpha_0]$ yields the pointwise bound
\begin{equation*}
\mathrm{KL}\big(m^{\ast}(\alpha)\|\pi\big)
\ \le\ \frac{1}{\alpha}\,F_\alpha\big(\mu^{\ast}_{\alpha_0}\big)
\ \le\ \frac{\gamma}{\alpha_0}\,\frac{\big(b(\mu^{\ast}_{\alpha_0})-a\big)^2}{2}
\ +\ \mathrm{KL}\big(\mu^{\ast}_{\alpha_0}\|\pi\big).
\end{equation*}
Both terms on the right are integrable (they appear in $F_{\alpha_0}(\mu^{\ast}_{\alpha_0})$), so they provide an $\alpha$–uniform integrable bound on $[\alpha_0/2,2\alpha_0]$. Dominated convergence then gives continuity of
\[
\alpha\ \longmapsto\ \mathbb E\!\int_0^T \mathrm{KL}\big(m_t^{\ast}(\alpha)\|\pi_t\big)\,dt
\]
at $\alpha_0$. Since $\alpha_0>0$ was arbitrary, continuity holds on $(0,\infty)$.
Write $G(\alpha):=\mathbb E\!\int_0^T \mathrm{KL}(m_t^{\ast}(\alpha)\|\pi_t)\,dt$. By the above, $G$ is continuous and nonincreasing on $(0,\infty)$. Moreover, for each $(t,\omega)$, $\theta(\alpha)\to 0$ as $\alpha\to\infty$ (since $\alpha\,\theta(\alpha)=\gamma(\psi(\theta(\alpha))-a)$ and $\psi$ is bounded), so $\mathrm{KL}(m_t^{\ast}(\alpha)\|\pi_t)\downarrow 0$. By the monotone convergence theorem, $G(\alpha)\downarrow 0$ as $\alpha\to\infty$. Therefore, for any $K\in(0,G(0+))$ there exists at least one $\alpha^{\ast}>0$ with $G(\alpha^{\ast})=K$; if $K\ge G(0+)$, the budget is slack and we set $\alpha^{\ast}=0$. If, in addition, $G$ is strictly decreasing (e.g. under the strictness condition stated above), then $\alpha^{*}$ is unique.
\end{proof}

\begin{remark}
In information geometry, the \emph{$I$-projection of $\pi\in\mathcal P(\Lambda)$ onto $\mathcal E\subset\mathcal P(\Lambda)$} is any
\[
\mu^{\star}\ \in\ \arg\min_{\mu\in\mathcal E}\ \mathrm{KL}(\mu\|\pi),
\]
whenever the minimum exists, where $(\Lambda,\mathcal B(\Lambda))$ is a measurable space and $\mathcal P(\Lambda)$ is the set of probability measures on it. When $\mathcal E$ is convex and suitably closed and $D_{\mathrm{KL}}(\cdot\|\pi)$ is finite on $\mathcal E$, the $I$-projection exists and is unique. In particular, if $\mathcal E$ is defined by linear moment constraints (e.g., $\int f_i\,d\mu=c_i$), then the $I$-projection $\mu^\star$ has an exponential-tilt density with respect to $\pi$.
\end{remark}

We collect the previous results into a single main theorem.  

\begin{theorem}\label{thm:main-solution}
Under Assumption~\ref{ass:data}, there exists $\alpha^{\ast}\ge 0$ and a $\mathbb G$–progressively measurable optimal control $m^{\ast}$ given by \eqref{eq:optimal-m-process}, where $\theta_t^{\ast}$ solves
\[
\psi_t(\theta)\ =\ \widehat a_t+\frac{\alpha^{\ast}}{\gamma}\,\theta,
\qquad
\psi_t(\theta)=\frac{\displaystyle\int \rho_t^\lambda e^{-\theta\rho_t^\lambda}\,\pi_t(d\lambda)}{\displaystyle\int e^{-\theta\rho_t^\lambda}\,\pi_t(d\lambda)}.
\]
If $\alpha^{\ast}>0$, then for a.e.\ $(t,\omega)$ the equation has a \emph{unique} solution $\theta_t^{\ast}(\omega)$, and the constraint is binding:
\[
\mathbb E\!\int_0^T \mathrm{KL}(m_t^{\ast}\|\pi_t)\,dt=K.
\]
If $\alpha^{\ast}=0$, then any adapted $m$ with $\bar\rho_t(m_t)=\widehat a_t$ a.e.\ and $\mathcal D_{\mathrm{KL}}(m\|\pi)\le K$ is optimal. In particular, whenever $\widehat a_t\in\big[\mathrm{ess\,inf}_\lambda \rho_t^\lambda,\ \mathrm{ess\,sup}_\lambda \rho_t^\lambda\big]$ a.s., one may select the (possibly nonunique) KL $I$–projection given by any solution $\theta_t$ of $\psi_t(\theta_t)=\widehat a_t$ as a canonical adapted minimizer.
\end{theorem}

\begin{proof}
Combine Theorem~\ref{thm:existence and scalar}, Propositions~\ref{prop:gibbs-solution} and \ref{prop:adaptedness}, Lemmas~\ref{lem:unique-theta} and \ref{lem:KL-budget}.
\end{proof}

Finally, we check that relaxing the constraints produces the filtered price process. 

\begin{theorem}\label{thm:large-budget-collapse}
Assume \ref{ass:data} and \ref{ass:attainability}. For $K>0$, let $m^{\ast,K}$ denote an optimal solution to Problem~\ref{prob:KL} and write
\[
\psi_t^{(K)}\ :=\ \bar\rho_t\big(m_t^{\ast,K}\big)\ =\ \int_{\Lambda}\rho_t^\lambda\,m_t^{\ast,K}(d\lambda).
\]
Let $\widehat S$ solve \eqref{eq:filtered-dynamics}, and for a fixed $\beta\in[0,1]$ let $\widetilde S^{(K)}$ solve
\begin{equation}\label{eq:synthetic-K}
d\widetilde S^{(K)}_t\ =\ \Big((1-\beta)\widehat a_t+\beta\,\psi_t^{(K)}\Big)\,dt\ +\ \widehat\sigma_t\,d\widehat W_t,
\qquad \widetilde S^{(K)}_0=\widehat S_0.
\end{equation}
Then, as $K\uparrow\infty$,
\begin{align}
\mathcal L\big(m^{\ast,K}\big)\ =\ \frac{\gamma}{2}\,\E\!\int_0^T\!\big(\psi_t^{(K)}-\widehat a_t\big)^2\,dt &\;\longrightarrow\; 0, \label{eq:Lto0}\\
\sup_{t\in[0,T]}\big|\widetilde S^{(K)}_t-\widehat S_t\big|
\;\le\; \beta\int_0^T\big|\psi_t^{(K)}-\widehat a_t\big|\,dt
&\;\xrightarrow[\ K\to\infty\ ]{\ \ \Prob\ \ }\; 0,\label{eq:collapse-prob}
\end{align}
and, in particular,
\[
\E\!\left[\sup_{t\in[0,T]}\big|\widetilde S^{(K)}_t-\widehat S_t\big|\right]
\ \le\ \beta\sqrt{\tfrac{2T}{\gamma}}\;\mathcal L\big(m^{\ast,K}\big)^{1/2}
\ \xrightarrow[\ K\to\infty\ ]{}\ 0.
\]
\end{theorem}

\begin{proof}[Proof sketch]
By complementary slackness \eqref{eq:comp-slackness}, the optimal multiplier satisfies $\alpha^{\ast}(K)\downarrow 0$ as $K\uparrow\infty$. Hence the constrained optimizers $m^{\ast,K}$ minimize
\[
\mathbb E\!\int_0^T\!\Big[\tfrac{\gamma}{2}\big(\bar\rho_t(m_t)-\widehat a_t\big)^2+\alpha^{\ast}(K)\,\KL(m_t\|\pi_t)\Big]\,dt.
\]
Let $m^0$ be as in Proposition~\ref{prop:zero-loss-benchmark}; using it as a competitor yields $\mathcal L(m^{\ast,K})\le \alpha^{\ast}(K)\, \E\!\int_0^T\!\KL(m^0_t\|\pi_t)\,dt$. Since $\alpha^{\ast}(K)\to 0$, this forces \eqref{eq:Lto0}. For \eqref{eq:collapse-prob}, note that $\widetilde S^{(K)}-\widehat S$ solves $d(\widetilde S^{(K)}-\widehat S)_t=\beta(\psi_t^{(K)}-\widehat a_t)\,dt$, therefore the pathwise bound and the $L^1$–bound follow by Cauchy–Schwarz.
\end{proof}

\subsection{Example: Uniform Prior and Affine Proposal Drifts}
\label{subsec:example-affine}

Fix $t$ and suppose $\pi_t$ is uniform on $[0,1]$. To ensure relaxing the constraint leads to matching the observable drift, define the expert proposal as a deviation from the filtered drift:
\[
\rho_t^\lambda \;=\; \widehat a_t \;+\; c_1(t)\,\lambda,\qquad c_1(t)>0,\ \ \lambda\in[0,1],
\]
and write $c_1=c_1(t)$ for brevity. Then
\begin{align}
\label{eq:Z-affine}
Z_t(\theta)
&=\int_0^1 e^{-\theta(\widehat a_t+c_1\lambda)}\,d\lambda
= e^{-\theta \widehat a_t}\,\frac{1-e^{-\theta c_1}}{\theta c_1},\\[0.2em]
\label{eq:psi-affine}
\psi_t(\theta)
&=\frac{\int_0^1 (\widehat a_t+c_1\lambda)\,e^{-\theta(\widehat a_t+c_1\lambda)}\,d\lambda}{\int_0^1 e^{-\theta(\widehat a_t+c_1\lambda)}\,d\lambda}
= \widehat a_t+c_1\left(\frac{1}{\theta c_1}-\frac{1}{e^{\theta c_1}-1}\right),
\end{align}
with $\psi_t(0)=\widehat a_t+c_1/2$ and $\psi_t(\theta)\downarrow \widehat a_t$ as $\theta\to\infty$.

For any $\alpha>0$ there is a unique $\theta_t^{\ast}$ solving
\begin{equation}
\label{eq:theta-affine}
c_1\left(\frac{1}{\theta c_1}-\frac{1}{e^{\theta c_1}-1}\right)=\frac{\alpha}{\gamma}\,\theta,
\end{equation}
and the optimal kernel \eqref{eq:optimal-m-process} has density
\begin{equation*}
m_t^{\ast}(d\lambda)=\frac{\theta_t^{\ast} c_1\,e^{-\theta_t^{\ast} c_1\,\lambda}}{1-e^{-\theta_t^{\ast} c_1}}\,d\lambda\quad\text{on }[0,1].
\end{equation*}
The pointwise KL reads
\begin{equation*}
\mathrm{KL}(m_t^{\ast}\|\pi_t)
=-\theta_t^{\ast}\,\psi_t(\theta_t^{\ast})-\log Z_t(\theta_t^{\ast})
=-1+\frac{\theta_t^{\ast} c_1}{e^{\theta_t^{\ast} c_1}-1}-\log\!\big(1-e^{-\theta_t^{\ast} c_1}\big)+\log(\theta_t^{\ast} c_1),
\end{equation*}
and the global constraint uniquely selects $\alpha^{\ast}>0$ solving
$\mathbb E\!\int_0^T \mathrm{KL}(m_t^{\ast}\|\pi_t)\,dt=K$,
or returns $\alpha^{\ast}=0$ if the constraint is slack. Since $\theta\mapsto \mathrm{KL}(m_t^{\ast}\|\pi_t)$ is strictly increasing for $\theta>0$, $K\uparrow$ implies $\theta_t^{\ast}\uparrow$ and thus $\psi_t(\theta_t^{\ast})\downarrow \widehat a_t$, so the aggregated mean (and hence the synthetic drift) converges to the filtered drift as $K\to\infty$. More precisely, as $\theta\uparrow\infty$, the $m_t^{\ast}(d\lambda)$ converges weakly to $\delta_{\{\lambda=0\}}$, the zero–loss selector $m^0$, and $\psi(\theta)\downarrow \widehat a$. Therefore, for the synthetic model
\[
d\widetilde S_t\ =\ \big((1-\beta)\widehat a_t+\beta\,\psi_t(\theta^{\ast}_t)\big)dt\ +\ \widehat\sigma_t\,d\widehat W_t,
\]
the drift collapses monotonically to $\widehat a_t$ as $K\uparrow\infty$, and $\widetilde S$ converges to the observable filter $\widehat S$ in the sense of Theorem~\ref{thm:large-budget-collapse}.

\begin{remark} If $\widehat a_t$ equals the essential infimum/supremum of $\rho_t^\lambda$ (e.g.\ here the infimum at $\lambda=0$ when $c_1>0$), then $\theta_t^{\ast}=\pm\infty$ in \eqref{eq:theta-affine} and $m_t^{\ast}$ concentrates on the corresponding extremal set; this is obtained as a limit of the formulas above.
\end{remark}

\subsection{Simulation of the True, Observed and Opinion-Biased Price Processes}

On a filtered probability space, the \emph{true} log–price \(X\) is given by
\begin{equation*}\label{eq:true}
dX_t \;=\; a_t\,dt \;-\; \tfrac12\,\sigma^2\,dt \;+\; \sigma\,dW_t^{S},\qquad S_t=S_0e^{X_t}.
\end{equation*}
The trader observes a drift–signal (used in the code for filtering)
\begin{equation*}\label{eq:obs}
dY_t \;=\; a_t\,dt \;+\; R^{1/2}\,dB_t,
\end{equation*}
and forms the filtered drift \(\widehat a_t:=\E[a_t\mid\mathcal F_t]\). The \emph{filtered} price admits the innovation representation
\begin{equation*}\label{eq:filtered}
d\widehat S_t \;=\; \widehat a_t\,dt \;+\; \widehat\sigma_t\,d\widehat W_t,
\end{equation*}
with \(\widehat W\) an \(\mathbb F\)–Brownian motion and \(\widehat\sigma_t\ge 0\).
Given an aggregation weight \(\beta\in[0,1]\) and an aggregated mean \(\psi_t\), the \emph{synthetic} price is
\begin{equation}\label{eq:synthetic}
d\widetilde S_t \;=\; \Big((1-\beta)\,\widehat a_t+\beta\,\psi_t\Big)dt \;+\; \widehat\sigma_t\,d\widehat W_t,
\qquad \widetilde S_0=\widehat S_0 .
\end{equation}

At time \(t\), experts \(\lambda\in[0,1]\) propose their individual bias terms
\begin{equation}\label{eq:rho}
\rho_t^\lambda \;=\; \widehat a_t \;+\; c_1(t)\,\lambda,\qquad c_1(t)>0,
\end{equation}
with prior \(\pi_t\) on \([0,1]\), where \(\pi_t=\mathrm{Beta}(a_\pi,b_\pi)\). For \(\theta\in\mathbb R\) the Gibbs measure optimizer is
\begin{equation*}\label{eq:gibbs}
m_t^\ast(d\lambda)\;=\;\frac{e^{-\theta\,\rho_t^\lambda}}{Z_t(\theta)}\,\pi_t(d\lambda),
\qquad
Z_t(\theta):=\int_0^1 e^{-\theta\,\rho_t^\lambda}\,\pi_t(d\lambda),
\end{equation*}
and the \emph{aggregated mean} $\psi_t(\theta)\;:=\;\int_0^1 \rho_t^\lambda\,m_t^\ast(d\lambda)
\;=\;-\frac{d}{d\theta}\log Z_t(\theta)$.

For the Beta prior \(\pi_t=\mathrm{Beta}(a_\pi,b_\pi)\) and \(\rho_t^\lambda=\widehat a_t+c_1\lambda\), letting \(u:=-\theta c_1\) and \(M(u):={}_1F_1(a_\pi;a_\pi+b_\pi;u)\), the confluent hypergeometric function,
\begin{equation*}\label{eq:beta-closed}
Z_t(\theta)=e^{-\theta \widehat a_t}M(u),\qquad 
\psi_t(\theta)=\widehat a_t+c_1\frac{a_\pi}{a_\pi+b_\pi}\frac{{}_1F_1(a_\pi{+}1;a_\pi{+}b_\pi{+}1;u)}{{}_1F_1(a_\pi;a_\pi{+}b_\pi;u)} .
\end{equation*}
The pointwise KL at the optimizer is
\begin{equation*}\label{eq:KL}
\KL\!\big(m_t^\ast\|\pi_t\big)\;=\;-\theta\,\psi_t(\theta)-\log Z_t(\theta)
\;=\; -\theta c_1\,\frac{a_\pi}{a_\pi+b_\pi}\frac{{}_1F_1(a_\pi{+}1;a_\pi{+}b_\pi{+}1;u)}{{}_1F_1(a_\pi;a_\pi{+}b_\pi;u)}-\log M(u).
\end{equation*}
\(\theta=\theta(K)\) is selected from
\(\E\!\int_0^T\KL(m_t^\ast\|\pi_t)\,dt=K\).

From \eqref{eq:psi} and \(Z_t''/Z_t-(Z_t'/Z_t)^2=\mathrm{Var}_{m_t^\ast}(\rho_t^\lambda)\),
\begin{equation}\label{eq:mono}
\frac{d}{d\theta}\psi_t(\theta)\;=\;-\mathrm{Var}_{m_t^\ast}(\rho_t^\lambda)\ \le 0,
\qquad
\frac{d}{d\theta}\KL\big(m_t^\ast\|\pi_t\big)\;=\;\theta\,\mathrm{Var}_{m_t^\ast}(\rho_t^\lambda)\ \ge 0.
\end{equation}
Hence \(K_1<K_2 \Rightarrow \theta(K_1)<\theta(K_2)\Rightarrow
\psi_t\big(\theta(K_1)\big)>\psi_t\big(\theta(K_2)\big)\).
With the choice \eqref{eq:rho}, \(\psi_t(\theta)\downarrow \widehat a_t\) as \(\theta\uparrow\infty\),
so the synthetic drift in \eqref{eq:synthetic} converges monotonically to \(\widehat a_t\) when \(K\to\infty\)
(\(\widetilde S\) collapses to \(\widehat S\)).

\begin{center}
\begin{tabular}{@{}rcccc@{}}
\hline
$K$ & $\theta$ & $\alpha$ & $\mathrm{KL}/T$ & $\delta_{\text{shift}}$ \\
\hline
$0.01$ & $1.013811$ & $0.23388$ & $0.010000$ & $0.237110$ \\
$0.5$  & $9.486238$ & $0.0138132$ & $0.500000$ & $0.131035$ \\
$5$    & $192.478831$ & $5.27449\times 10^{-5}$ & $5.000000$ & $0.010152$ \\
$20$   & $364373.533564$ & $1.50637\times 10^{-11}$ & $20.000000$ & $5\times 10^{-6}$ \\
\hline
\end{tabular}
\end{center}

\begin{figure}[t]
  \centering
  \includegraphics[width=\linewidth,height=0.9\textheight,keepaspectratio]{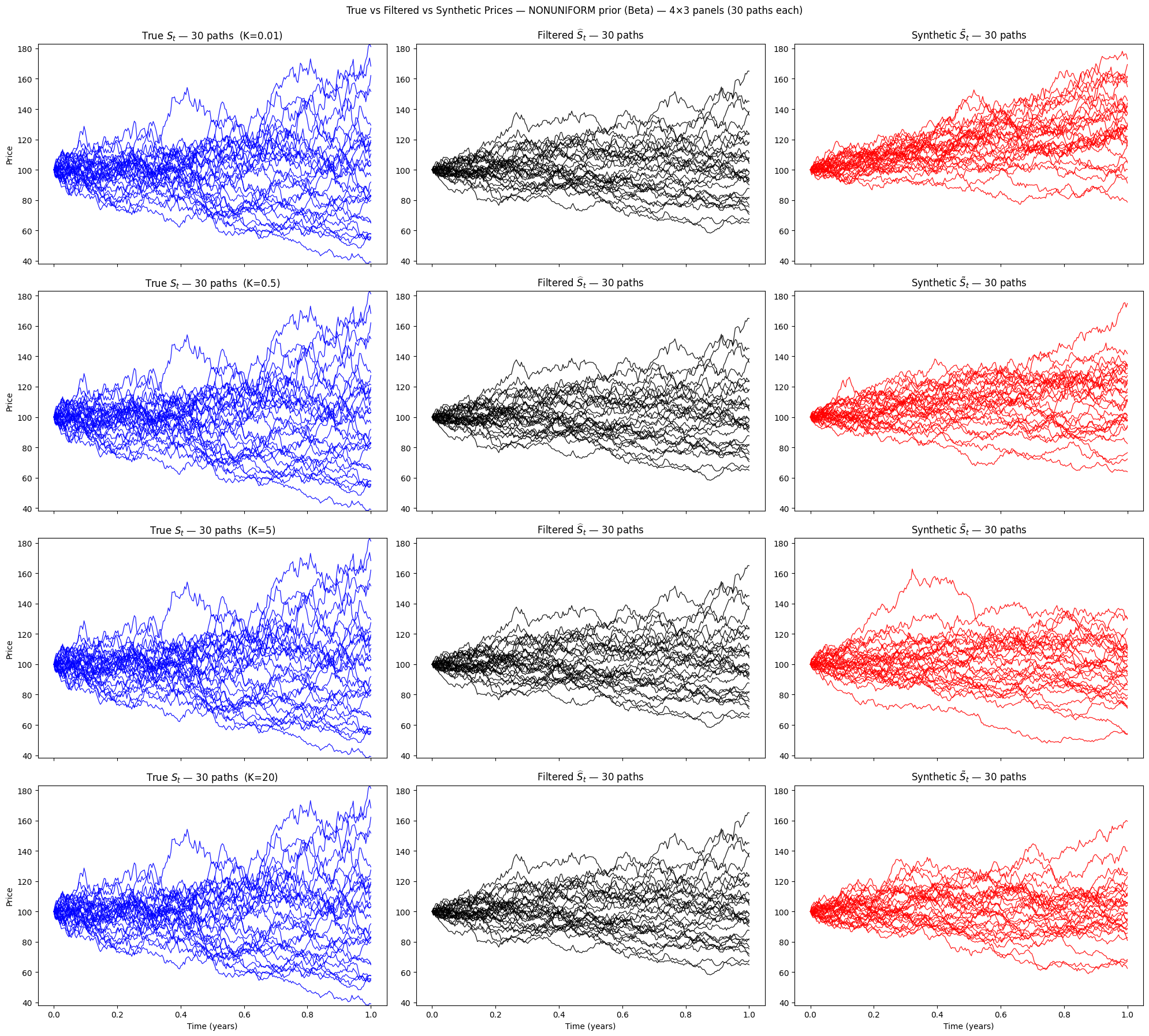}
  \caption{Rows: increasing constraints \(K_1<K_2<K_3<K_4\). Columns: true \(S_t\) (blue), filtered \(\widehat S_t\) (black), synthetic \(\widetilde S_t\) (red). Common $y$–scale across all panels. Average $\mathrm{Corr}(a,\widehat a)$ across 30 paths is $0.8501$.}
  \label{fig:Sim3}
\end{figure}

\clearpage

\begin{color}{black}
\section{Conclusion}
We have studied three successively more specialized models of financial markets under information constraints with traders with differential beliefs. Our analysis has shown that under a fairly natural compatibility condition, increasing information leads to an efficient market as defined in \cite{JarrowLarsson} provided the increase in information is uniform across different traders. The more specialized model incorporating an individual trader's biases introduced a novel way of measuring the impact of a trader's intuitive sense of ambiguity regarding the true value of a partially observed price process. Again, under some natural assumptions, increasing information leads to a decrease in this ambiguity and hence shrinking biases. Finally, we solved a stochastic optimal control problem for a trader seeking positive alphas as defined in \cite{JarrowProtter}, i.e. an arbitrage opportunity or a dominated asset. From a mathematical point of view, our optimal solution is formally similar to well-known results in information theory and information geometry. All our theoretical results were implemented in simulations and hence are well-suited for practical applications to asset pricing in markets with significant information constraints and price evolutions affected by differential beliefs of market participants. 
\end{color}


\appendix                          
\section*{Appendix A}

Let $(\Lambda,d)$ be a compact metric space and let $\mathcal P(\Lambda)$ denote the space of Borel probability measures on $\Lambda$. Write $W_2$ for the 2–Wasserstein distance on $\mathcal P(\Lambda)$ induced by $d$. Since $\Lambda$ is compact, every probability measure has finite second moment, so $\mathcal P(\Lambda)=\mathcal P_2(\Lambda)$, and the $W_2$–topology coincides with the topology of weak convergence. In particular, $(\mathcal P(\Lambda),W_2)$ is compact and its Borel $\sigma$–algebra, denoted as $\mathcal B(\mathcal P(\Lambda))$, agrees with the Borel $\sigma$–algebra $\mathcal B_w(\mathcal P(\Lambda))$ for the weak topology.

A \emph{flow of probability kernels} on $\Lambda$ is a map $(t,\omega)\longmapsto m_t(\omega)\in \mathcal P(\Lambda)$ that is $(\mathcal B([0,T])\otimes\mathcal F)$–measurable as a map from $[0,T]\times\Omega$ into the Polish space $(\mathcal P(\Lambda),\mathcal B(\mathcal P(\Lambda)))$.

We say $m$ is \emph{$\mathbb F$–adapted} if, for each fixed $t\in[0,T]$, the map $\omega\longmapsto m_t(\omega)$ is $(\mathcal F_t,\mathcal B(\mathcal P(\Lambda))$–measurable.

Let $\mathcal P(\mathbb F)$ denote the progressive $\sigma$–algebra on $[0,T]\times\Omega$ so that, for each $t$, the restriction to $[0,t]\times\Omega$ coincides with $\mathcal B([0,t])\otimes\mathcal F_t$. We say $m$ is \emph{progressively measurable} if $(t,\omega)\longmapsto m_t(\omega)$ is $\mathcal P(\mathbb F)$–measurable as a map into $(\mathcal P(\Lambda),\mathcal B(\mathcal P(\Lambda)))$.

\begin{lemma}\label{lem:kernel-equivalences}
Let $\Lambda$ be Polish and let $(\Omega,\mathcal F)$ be a measurable space. For a map
$\mu:\Omega\to\mathcal P(\Lambda)$ with $\mathcal P(\Lambda)$ endowed with the weak topology and its Borel $\sigma$–algebra $\mathcal B_w(\mathcal P(\Lambda))$, the following are equivalent:
\begin{enumerate}
\item[(i)] $\mu$ is $(\mathcal F,\mathcal B_w(\mathcal P(\Lambda)))$–measurable;
\item[(ii)] for all $\varphi\in C_b(\Lambda)$, the map $\omega\mapsto \int_\Lambda \varphi(\lambda)\,d\mu(\omega, d\lambda)$ is $\mathcal F$–measurable;
\item[(iii)] for all open $G\subset\Lambda$, the map $\omega\mapsto \mu(\omega, G)$ is $\mathcal F$–measurable;
\item[(iv)] for all $B\in\mathcal B(\Lambda)$, the map $\omega\mapsto \mu(\omega, B)$ is $\mathcal F$–measurable.
\end{enumerate}
\end{lemma}

\begin{proof}
    Standard. See Kallenberg, \cite{Kallenberg}, Lemma 1.40. 
\end{proof}

\begin{lemma}\label{lem:bar-rho-meas}
Assume (A1). Then
\begin{enumerate}
\item If $m$ is adapted, then for each fixed $t$, the map
$\omega\mapsto \bar\rho_t(m_t)(\omega):=\int_\Lambda \rho_t^\lambda(\omega)\,m_t(\omega,d\lambda)$ is
$\mathcal F_t$–measurable.
\item If $m$ is progressively measurable, then the map
$(t,\omega)\mapsto \bar\rho_t(m_t)(\omega)$ is $\mathcal P(\mathbb F)$–measurable (hence progressively
measurable as a real–valued process).
\end{enumerate}
\end{lemma}

\begin{proof}
    Standard. Follows from Kallenberg, \cite{Kallenberg}, Lemma 1.41 by standard arguments. 
\end{proof}

\end{document}